\documentclass{preprint}
\usepackage{helvet}
\usepackage{courier}
\usepackage[T1]{fontenc}
\usepackage[a4paper]{geometry}
\usepackage[active]{srcltx}
\usepackage{color}
\usepackage{verbatim}
\usepackage{mathtools}
\usepackage{amsmath}
\usepackage{amsthm}
\usepackage{amssymb}
\usepackage{makeidx}
\makeindex
\usepackage[pdfusetitle,
 bookmarks=true,bookmarksnumbered=true,bookmarksopen=true,bookmarksopenlevel=2,
 breaklinks=false,pdfborder={0 0 1},backref=page,colorlinks=false]
 {hyperref}
\hypersetup{
 unicode=false, linkcolor=black, citecolor=black, urlcolor=blue, filecolor=blue, pdfpagelayout=OneColumn, pdfnewwindow=true, pdfstartview=XYZ, plainpages=false}

\makeatletter
\numberwithin{equation}{section}
\theoremstyle{plain}
\newtheorem{thm}{\protect\theoremname}[section]
\theoremstyle{plain}
\newtheorem{prop}[thm]{\protect\propositionname}
\theoremstyle{plain}
\newtheorem{lem}[thm]{\protect\lemmaname}
\theoremstyle{definition}
\newtheorem{defn}[thm]{\protect\definitionname}

\usepackage{color}

 \let\myTOC\tableofcontents
 \renewcommand\tableofcontents{%
   \pdfbookmark[1]{Contents}{}
   \myTOC
   \cleardoublepage
   \pagenumbering{arabic} }

\global\long\def\foreignlanguage#1#2{#2}%
\global\long\def\selectlanguage#1{}%

\allowdisplaybreaks

\makeatother

\providecommand{\definitionname}{Definition}
\providecommand{\lemmaname}{Lemma}
\providecommand{\propositionname}{Proposition}
\providecommand{\theoremname}{Theorem}

\begin{document}
\title{The conditional probabilities and the empirical laws in a free scalar
QFT in curved spacetime}
\author{Hideyasu Yamashita}
\institute{Division of Liberal Arts and Sciences, Aichi-Gakuin University\
\email{yamasita@dpc.aichi-gakuin.ac.jp}}

\date{\today}

\maketitle
\newcommand{\dcolor}{\definecolor{note_fontcolor}{rgb}{0.1, 0.0, 0.8}}
\definecolor{HYnote_fontcolor}{rgb}{0.2, 0.0, 0.2}
\definecolor{hycolor}{rgb}{0.3, 0.0, 0.3}
\newcommand{\hyc}{\color{hycolor}}
\newenvironment{HYnote}
 {\textcolor{note_fontcolor}\bgroup\ignorespaces}
  {\ignorespacesafterend\egroup} 

\newenvironment{trivenv}
  {\bgroup\ignorespaces}
  {\ignorespacesafterend\egroup}

\newcommand{\displabel}[1]{}

\newcommand{\hidable}[3]{#2}
\newcommand{\hidea}[1]{{#1}}
\newcommand{\hideb}[1]{{#1}}
\newcommand{\hidec}[1]{{#1}}
\newcommand{\hidep}[1]{{#1}}
\renewcommand{\hidec}[1]{}
\renewcommand{\hidep}[1]{}

\newcommand{\thlab}[1]{{\tt [#1]}}

\newcommand{\black}{\color{black}}

\global\long\def\N{\mathbb{N}}%
\global\long\def\C{\mathbb{C}}%
\global\long\def\Z{\mathbb{Z}}%
 
\global\long\def\R{\mathbb{R}}%
 
\global\long\def\im{\mathrm{i}}%

\global\long\def\di{\partial}%
 
\global\long\def\d{{\rm d}}%

\global\long\def\ol#1{\overline{#1}}%
\global\long\def\ul#1{\underline{#1}}%
\global\long\def\ob#1{\overbrace{#1}}%

\global\long\def\ov#1{\overline{#1}}%

\global\long\def\then{\Rightarrow}%
 
\global\long\def\Then{\Longrightarrow}%

\global\long\def\N{\mathbb{N}}%
\global\long\def\C{\mathbb{C}}%
\global\long\def\Z{\mathbb{Z}}%
 
\global\long\def\R{\mathbb{R}}%
 
\global\long\def\im{\mathrm{i}}%

\global\long\def\di{\partial}%
 
\global\long\def\d{{\rm d}}%

\global\long\def\ol#1{\overline{#1}}%
\global\long\def\ul#1{\underline{#1}}%
\global\long\def\ob#1{\overbrace{#1}}%

\global\long\def\ov#1{\overline{#1}}%

\global\long\def\then{\Rightarrow}%
 
\global\long\def\Then{\Longrightarrow}%

\global\long\def\cA{\mathcal{A}}%
\global\long\def\cB{\mathcal{B}}%
 
\global\long\def\cC{\mathcal{C}}%
 
\global\long\def\cD{\mathcal{D}}%
\global\long\def\cE{\mathcal{E}}%
 
\global\long\def\cF{\mathcal{F}}%
 
\global\long\def\cG{{\cal G}}%
 
\global\long\def\cH{\mathcal{H}}%
 
\global\long\def\cI{\mathcal{I}}%
 
\global\long\def\cJ{\mathcal{J}}%
\global\long\def\cK{\mathcal{K}}%
 
\global\long\def\cL{\mathcal{L}}%
 
\global\long\def\cM{\mathcal{M}}%
 
\global\long\def\cN{\mathcal{N}}%
 
\global\long\def\cO{\mathcal{O}}%
 
\global\long\def\cP{\mathcal{P}}%
 
\global\long\def\cQ{\mathcal{Q}}%
 
\global\long\def\cR{\mathcal{R}}%
 
\global\long\def\cS{\mathcal{S}}%
 
\global\long\def\cT{\mathcal{T}}%
 
\global\long\def\cU{\mathcal{U}}%
 
\global\long\def\cV{\mathcal{V}}%
 
\global\long\def\cW{\mathcal{W}}%
\global\long\def\cX{\mathcal{X}}%
 
\global\long\def\cY{\mathcal{Y}}%
 
\global\long\def\cZ{\mathcal{Z}}%

\global\long\def\scA{\mathscr{A}}%
\global\long\def\scB{\mathscr{B}}%
\global\long\def\scC{\mathscr{C}}%
\global\long\def\scD{\mathscr{D}}%
 
\global\long\def\scE{\mathscr{E}}%
 
\global\long\def\scF{\mathscr{F}}%
 
\global\long\def\scG{\mathscr{G}}%
 
\global\long\def\scH{\mathscr{H}}%
 
\global\long\def\scI{\mathscr{I}}%
 
\global\long\def\scJ{\mathscr{J}}%
 
\global\long\def\scK{\mathscr{K}}%
 
\global\long\def\scL{\mathscr{L}}%
 
\global\long\def\scM{\mathscr{M}}%
 
\global\long\def\scN{\mathscr{N}}%
 
\global\long\def\scO{\mathscr{O}}%
 
\global\long\def\scP{\mathscr{P}}%
 
\global\long\def\scR{\mathscr{R}}%
\global\long\def\scS{\mathscr{S}}%
 
\global\long\def\scT{\mathscr{T}}%
 
\global\long\def\scU{\mathscr{U}}%
 
\global\long\def\scW{\mathscr{W}}%
\global\long\def\scZ{\mathscr{Z}}%

\global\long\def\bbA{\mathbb{A}}%
 
\global\long\def\bbB{\mathbb{B}}%
 
\global\long\def\bbD{\mathbb{D}}%
 
\global\long\def\bbE{\mathbb{E}}%
 
\global\long\def\bbF{\mathbb{F}}%
 
\global\long\def\bbG{\mathbb{G}}%
 
\global\long\def\bbI{\mathbb{I}}%
 
\global\long\def\bbJ{\mathbb{J}}%
 
\global\long\def\bbK{\mathbb{K}}%
 
\global\long\def\bbL{\mathbb{L}}%
 
\global\long\def\bbM{\mathbb{M}}%
 
\global\long\def\bbP{\mathbb{P}}%
 
\global\long\def\bbQ{\mathbb{Q}}%
 
\global\long\def\bbT{\mathbb{T}}%
 
\global\long\def\bbU{\mathbb{U}}%
 
\global\long\def\bbX{\mathbb{X}}%
 
\global\long\def\bbY{\mathbb{Y}}%
\global\long\def\bbW{\mathbb{W}}%

\global\long\def\bbOne{1\kern-0.7ex  1}%
 %defined as 1\kern-0.7ex1

\renewcommand{\bbOne}{\mathbbm{1}}

\global\long\def\bB{\mathbf{B}}%
 
\global\long\def\bG{\mathbf{G}}%
 
\global\long\def\bH{\mathbf{H}}%
\global\long\def\bS{\boldsymbol{S}}%
 
\global\long\def\bT{\mathbf{T}}%
 
\global\long\def\bX{\mathbf{X}}%
\global\long\def\bY{\mathbf{Y}}%
\global\long\def\bW{\mathbf{W}}%
 
\global\long\def\boT{\boldsymbol{T}}%

\global\long\def\fraka{\mathfrak{a}}%
 
\global\long\def\frakb{\mathfrak{b}}%
 
\global\long\def\frakc{\mathfrak{c}}%
 
\global\long\def\frake{\mathfrak{e}}%
 
\global\long\def\frakf{\mathfrak{f}}%
 
\global\long\def\fg{\mathfrak{g}}%
 
\global\long\def\frakh{\mathfrak{h}}%
 
\global\long\def\fraki{\mathfrak{i}}%
\global\long\def\frakk{\mathfrak{k}}%
 
\global\long\def\frakl{\mathfrak{l}}%
 
\global\long\def\frakm{\mathfrak{m}}%
 
\global\long\def\frakn{\mathfrak{n}}%
 
\global\long\def\frako{\mathfrak{o}}%
 
\global\long\def\frakp{\mathfrak{p}}%
 
\global\long\def\frakq{\mathfrak{q}}%
 
\global\long\def\fraks{\mathfrak{s}}%
 
\global\long\def\fs{\mathfrak{s}}%
 
\global\long\def\fraku{\mathfrak{u}}%
\global\long\def\frakz{\mathfrak{z}}%

\global\long\def\fA{\mathfrak{A}}%
 
\global\long\def\fB{\mathfrak{B}}%
 
\global\long\def\fC{\mathfrak{C}}%
 
\global\long\def\fD{\mathfrak{D}}%
 
\global\long\def\fF{\mathfrak{F}}%
 
\global\long\def\fG{\mathfrak{G}}%
 
\global\long\def\fK{\mathfrak{K}}%
 
\global\long\def\fL{\mathfrak{L}}%
 
\global\long\def\fM{\mathfrak{M}}%
 
\global\long\def\fP{\mathfrak{P}}%
 
\global\long\def\fR{\mathfrak{R}}%
 
\global\long\def\fS{\mathfrak{S}}%
\global\long\def\fT{\mathfrak{T}}%
 
\global\long\def\fU{\mathfrak{U}}%
 
\global\long\def\fX{\mathfrak{X}}%

\global\long\def\ssS{\mathsf{S}}%
\global\long\def\ssT{\mathsf{T}}%
 
\global\long\def\ssW{\mathsf{W}}%

\global\long\def\hM{\hat{M}}%

\global\long\def\rM{\mathrm{M}}%
\global\long\def\prj{\mathfrak{P}}%

{} 
\global\long\def\sy#1{{\color{blue}#1}}%

\global\long\def\magenta#1{{\color{magenta}#1}}%

% \global\long\def\symb#1{{\color{red}#1}}%
\global\long\def\symb#1{#1}%

\global\long\def\emhrb#1{\text{{\color{red}{\huge {\bf #1}}}}}%

\newcommand{\symbi}[1]{\index{$ #1$}{\color{red}#1}} 

{} 
% \global\long\def\SYM#1#2{\symb{#1}_{\##2}}%
\global\long\def\SYM#1#2{#1}%

\renewcommand{\SYM}[2]{\symb{#1}}

\newcommand{\usuji}{\color[rgb]{0.7,0.4,0.4}} \newcommand{\usu}{\color[rgb]{0.5,0.2,0.1}}
\newenvironment{Usuji} {\begin{trivlist}   \item \usuji }  {\end{trivlist}}
\newenvironment{Usu} {\begin{trivlist}   \item \usu }  {\end{trivlist}} 

\newcommand{\term}[1]{\textcolor[rgb]{0, 0, 1}{\bf #1}}
\newcommand{\termi}[1]{{\bf #1}}

\global\long\def\supp{{\rm supp}}%
\global\long\def\dom{\mathrm{dom}}%
\global\long\def\ran{\mathrm{ran}}%
 
\global\long\def\leng{\text{{\rm leng}}}%
 
\global\long\def\diam{\text{{\rm diam}}}%
 
\global\long\def\Leb{\text{{\rm Leb}}}%
 
\global\long\def\meas{\text{{\rm meas}}}%
\global\long\def\sgn{{\rm sgn}}%
 
\global\long\def\Tr{{\rm Tr}}%
 
\global\long\def\tr{\mathrm{tr}}%
 
\global\long\def\spec{{\rm spec}}%
 
\global\long\def\Ker{{\rm Ker}}%
 
\global\long\def\Lip{{\rm Lip}}%
 
\global\long\def\Id{{\rm Id}}%
 
\global\long\def\id{{\rm id}}%

\global\long\def\ex{{\rm ex}}%
 
\global\long\def\Pow{\mathsf{P}}%
 
\global\long\def\Hom{\mathrm{Hom}}%
 
\global\long\def\div{\mathrm{div}}%
 
\global\long\def\grad{\mathrm{grad}}%
 
\global\long\def\Lie{{\rm Lie}}%
 
\global\long\def\End{{\rm End}}%
 
\global\long\def\Ad{{\rm Ad}}%

\newcommand{\slim}{\mathop{\mbox{s-lim}}} %

\newcommand{\wlim}{\mathop{\mbox{w-lim}}}

\newcommand{\limsub}{\mathop{\mbox{\rm lim-sub}}}

\global\long\def\bboxplus{\boxplus}%

\renewcommand{\bboxplus}{\mathop{\raisebox{-0.8ex}{\text{\begin{trivenv}\LARGE{}$\boxplus$\end{trivenv}}}}}

\global\long\def\shuff{\sqcup\kern-0.3ex  \sqcup}%

\renewcommand{\shuff}{\shuffle}

\global\long\def\upha{\upharpoonright}%

\global\long\def\ket#1{|#1\rangle}%
 
\global\long\def\bra#1{\langle#1|}%

{} 
\global\long\def\lll{\vert\kern-0.25ex  \vert\kern-0.25ex  \vert}%
 \renewcommand{\lll}{{\vert\kern-0.25ex  \vert\kern-0.25ex  \vert}}

\global\long\def\biglll{\big\vert\kern-0.25ex  \big\vert\kern-0.25ex  \big\vert\kern-0.25ex  }%
 
\global\long\def\Biglll{\Big\vert\kern-0.25ex  \Big\vert\kern-0.25ex  \Big\vert}%

\newcommand{\iiia}[1]{{\left\vert\kern-0.25ex\left\vert\kern-0.25ex\left\vert #1
  \right\vert\kern-0.25ex\right\vert\kern-0.25ex\right\vert}}

\global\long\def\iii#1{\iiia{#1}}%

\global\long\def\Upa{\Uparrow}%
 
\global\long\def\Nor{\Uparrow}%

\newcommand{\vertt}{\kern-0.6ex\vert}
\renewcommand{\Nor}{[\kern-0.16ex ]}

\global\long\def\Prob{\mathbb{P}}%
\global\long\def\Var{\mathrm{Var}}%
\global\long\def\Cov{\mathrm{Cov}}%
\global\long\def\Ex{\mathbb{E}}%
{} %\newcommand{\F}{\mathbf{F}}
\global\long\def\Ae{{\rm a.e.}}%
 
\global\long\def\samples{\bOm}%

% misc %%%%%%%%%%%%%%%%%%%%%%%%%%%%%%%%%%%%%%%%%%%%%%%%%%%%%

\global\long\def\bOne{{\bf 1}}%

\global\long\def\Ten{\bullet}%
{} %

\global\long\def\TT{\intercal}%
 \renewcommand{\TT}{\mathsf{T}}

\global\long\def\trit{\vartriangle\!\! t}%

\global\long\def\p{\mathbf{p}}%
 
\global\long\def\q{\mathbf{q}}%
 
\global\long\def\bA{\mathbf{A}}%
 
\global\long\def\x{\mathbf{x}}%
 
\global\long\def\y{\mathbf{y}}%

\global\long\def\WICK#1{{\rm w}\{#1\}}%
 \renewcommand{\WICK}[1]{{:}#1{:}}

\global\long\def\ssD{\boldsymbol{D}}%
 \renewcommand{\ssD}{\mathsf{D}}

\global\long\def\ssK{\boldsymbol{K}}%
 \renewcommand{\ssK}{\mathsf{K}}

\global\long\def\ssH{\boldsymbol{H}}%
 \renewcommand{\ssH}{\mathsf{H}}

\global\long\def\ssN{\boldsymbol{N}}%
 \renewcommand{\ssN}{\mathsf{N}}

\global\long\def\ssR{\boldsymbol{R}}%
 \renewcommand{\ssR}{\mathsf{R}}

\global\long\def\ssP{\boldsymbol{P}}%
 \renewcommand{\ssP}{\mathsf{P}}

\global\long\def\bQ{\mathbf{Q}}%
 
\global\long\def\rF{\mathrm{F}}%
 
\global\long\def\bM{\mathbf{M}}%
 
\global\long\def\rE{\mathrm{E}}%
 
\global\long\def\bV{\mathbf{V}}%

\global\long\def\bE{\mathbf{E}}%

\global\long\def\sfS{\mathsf{S}}%
\global\long\def\sfQ{\mathsf{Q}}%
 
\global\long\def\bfS{{\bf S}}%

\global\long\def\Span{{\rm span}}%
\global\long\def\Inv{{\rm Inv}}%
\global\long\def\Borel{{\rm Borel}}%
\global\long\def\Lag{{\rm Lag}}%
\global\long\def\Tv{{\rm Tv}}%

\global\long\def\GreenOp{\mathsf{E}}%
\global\long\def\Sol{\mathsf{Sol}}%
\global\long\def\GL{{\rm GL}}%
\global\long\def\Sp{{\rm Sp}}%
\global\long\def\sfD{\mathsf{D}}%
\global\long\def\ann{{\rm ann}}%
\global\long\def\Ann{{\rm Ann}}%

\global\long\def\sperp{\angle}%
{} 
\global\long\def\bProj{\mathfrak{P}}%

\def\foreignlanguage#1#2{#2}

%\clearpage

\let\ruleorig=\rule
\renewcommand{\rule}{\noindent\ruleorig}

\global\long\def\labelenumi{(\arabic{enumi})}%

\begin{abstract}
Unlike QFT in Minkowski spacetime (QFTM), QFT in curved spacetime
(QFTCS) suffers from a conceptual obscurity on the empirical (experimentally
verifiable/falsifiable) laws. We propose to employ the notion of \emph{prior
conditional probabilities} to describe a part of the empirical laws
of QFTCS. This is interpreted as a quantum conditional probability
\emph{without no information on the initial state}. Hence this notion
is expected to be free from the inevitable vagueness of the empirical
meaning of quantum states in QFTCS. More generally in quantum physics,
this notion seems free from the conceptual problems on state reductions.
We confine ourselves to the probabilistic laws of the free scalar
fields (Klein\textendash Gordon fields) in curved spacetime, which
require some reconsideration on the empirical meaning of the canonical
commutation relation (CCR). We give some examples of empirical laws
in terms of prior conditional probabilities, concerning the CCR and
the free scalar QFTCS.
\end{abstract}

\section{Introduction}

First, consider the following simple question:%
{} How can \emph{physical/empirical} laws%
{} be described by a QFT on a Minkowski spacetime (QFTM)? A popular
answer will be the following: They are described by the \emph{S-matrix}.
One may give another more general but more indirect answer: All the
empirical laws are essentially derived from the \emph{Wightman functions}
(i.e., the vacuum expectation values).

On the other hand, in a QFT in curved spacetime (QFTCS) (see \cite{Wal1994,BFV2003,BF2009,HW2015,KM2015,FR2016}
and references therein), neither S-matrices nor Wightman functions
are defined on a generic spacetime. The undefinability of Wightman
functions follows from the undefinability of the vacuum; In QFTCS,
generally there is no unique distinguished %
state such as the vacuum.%
{} Thus we have a fundamental question: How can empirical laws be described
by a QFTCS? What causes the predictive power of a QFTCS?

Fewster and Verch \cite{FV2020} elaborate a formulation of measurement
processes in QFTCS, which is expected to make the empirical meaning
of QFTCS clearer. However, the above fundamental question seems to
remain unanswered there; For example, what is meant empirically/experimentally
when saying ``{[}a{]} probe is prepared in a state $\sigma$'' \cite[Sec.3.2]{FV2020}
in QFTCS? We know that in QFTM we can obtain the vacuum state which
is interpreted to be ``prepared by Nature'', and the states which
are given by local excitations of the vacuum. However, in QFTCS we
do not know precisely what is a ``preparable state''. 

The \emph{Unruh effect} is known as one of the most prominent results
of QFTCS. However, it seems to remain an unsolved problem whether
the Unruh effect is an empirical law, that is, whether we can verify/falsify
the Unruh effect experimentally. See \cite{CHM2007,Ear2011} and references
therein. Thus we see that the notions of empirical laws, as well as
those of observability and verifiability, have been rather obscure
in QFTCS. Note that the conceptual clarity of such notions does not
follow necessarily from the mathematical rigor of the theory. In fact,
the rigorous proof of the Unruh effect was found over forty years
ago by Sewell \cite{Sew1982}, based on the Bisognano\textendash Wichmann
theorem \cite{BW1975} and the Tomita\textendash Takesaki theory (e.g.~\cite{BR87,Tak2003}).
Although the question on the empirical meaning of the Unruh effect
is one of the motivations of this article, we will not arrive at its
answer here. In this article we confine ourselves to the probabilistic
laws of the free scalar fields (Klein\textendash Gordon fields) in
curved spacetime, which require some reconsideration on the empirical
meaning of the canonical commutation relation (CCR).

In this paper, we propose to employ the notion of \emph{prior conditional
probabilities} to describe a part of the empirical laws. This is interpreted
as a quantum conditional probability \emph{without no information
on the initial state}. Hence this notion is expected to be free from
the inevitable vagueness of the empirical meaning of quantum states
in QFTCS. More generally in quantum physics, this notion seems free
from the conceptual problems on state reductions (or wave packet collapses);
Although some authors are struggling to resolve the conceptual discrepancy
between the state reductions and the relativistic causality, we emphasize
that prior conditional probabilities can be defined without referring
to state reductions. (Note that we will not mention whether the state
reductions are physically realistic phenomena or not.)

It is also emphasized that we will not attempt to describe \emph{all}
of the empirical laws of a QFTCS in such a way. For example, our method
does not seem suitable to describe the thermal effects such as the
Unruh/Hawking effects. Furthermore, we confine ourselves to the field
observables of a scalar QFTCS, and hence we shall not refer to the
non-linear observables such as the stress-energy tensor in this paper.

Next we mention the problems of describing physical symmetries of
the empirical laws. %
{} Because the vacuum is the only Poincar\'e invariant state in QFTM,
the covariance of a QFTM is considered to be fully described by the
covariance of vacuum expectations (Wightman functions). On the other
hand, a QFTCS is required to be \emph{generally covariant}, but first
we must consider the elementary question: what is the general covariance?
We encounter the problem to give a precise definition of the notion,
both mathematically and conceptually. See Norton \cite{Nor1993} for
the long history of the disputes on that notion. Note that there have
been possibly many authors who think that the general covariance is
a physically vacuous, somewhat tautological principle \cite[Sec.5]{Nor1993}.
Thus it is not quite evident whether the general covariance is an
empirical (experimentally verifiable/falsifiable) law. Nowadays we
have a rigorous definition of the general covariance in QFTCS due
to Brunetti, Fredenhagen \& Verch \cite{BFV2003}, and also the definition
of Hollands \& Wald \cite{HW2010}. However note again that the mathematical
rigor does not necessarily guarantee the conceptual clarity. If we
are not sure whether the general covariance of QFTCS is empirical,
on what grounds can we say that a QFTCS is a \emph{physical} theory?
Of course we will not be able to give the final answer to this difficult
fundamental problem in this article, but we expect that our method
of prior conditional probability can shed some light on it. This expectation
is based on some examples of the symplectic covariant descriptions
of empirical laws in terms of prior conditional probabilities given
in Subsec.~\ref{subsec:Covariant-reformulation} and Sec.~\ref{sec:Finite-dimensional-case},
and an example of Galilei covariant description in Subsec.~\ref{subsec:Galilei-covariant-free}.
Although we also would like to regard Theorem \ref{thm:KG3} as an
example of the generally covariant empirical law in a QFTCS, the above-mentioned
conceptual obscurity of general covariance prevents us from asserting
that.

The question on the empirical meaning of QFTCS%
{} arises even in the the simplest cases: free scalar fields (Klein\textendash Gordon
fields) in a curved spacetime. In the algebraic formulations of QFTCS,
roughly speaking, a free scalar field in a curved (globally hyperbolic)
spacetime $\cM$ is defined by the canonical commutation relation
(CCR) \cite{BGP2007,BF2009,HW2015,KM2015,BD2015}. That is, a free
scalar field is defined to be an abstract CCR algebra $\cA(V)$ \cite[Theorem 5.2.8]{BR97}
over a symplectic vector space $V$ of classical solutions of the
Klein\textendash Gordon equation on $\cM$. In other words, a free
scalar field $\phi$ in $\cM$ with a given mass $m\ge0$ has only
one (nontrivial) law: $[\phi(f),\phi(g)]=\im\sigma_{m}(f,g)\bOne$,
where $f,g$ are test functions, $\sigma_{m}$ is a symplectic form,
and $\bOne$ is the unit element. Hence the question on the empirical
meaning of a free scalar field on $\cM$ seems to boil down to the
same question on the CCR.%

In Section \ref{sec:Definitions-of-free}, we present a definition
of a free scalar field on a curved spacetime, and in Section \ref{sec:Some-empirical-laws}
we reconsider a quite elementary but fundamental question: What is
the empirical law (or the predictive power) of the CCR? 

In Section \ref{sec:Harmonic-oscillator}, we describe some empirical
laws in terms of prior conditional probabilities on harmonic oscillators,
viewed as a ``free scalar fields in $(0+1)$-dimensional spacetime''.
In Section \ref{sec:Finite-dimensional-case}, we consider any finite-dimensional
symplectic vector space $V$, and the CCR over $V$. We obtain an
explicit formula to calculate prior conditional probabilities on the
simplest cases. In Section \ref{sec:Free-scalar-field} we examine
general free scalar (Klein\textendash Gordon) field on a globally
hyperbolic spacetime, and show an explicit formula to give an explicit
formula on a prior conditional probabilities, derived from the result
of Section \ref{sec:Finite-dimensional-case}.

\section{Definitions of free scalar (Klein\textendash Gordon) fields on a
curved spacetime}

\label{sec:Definitions-of-free}

First recall the definitions of the abstract CCR algebra. A representation-independent
definition of a CCR algebra is as follows (e.g., \cite{KM2015}).
Let $V$ be a real linear space equipped with a nondegenerate symplectic
bilinear form $\sigma$. The \termi{CCR algebra} over $V$ is the
unital {*}-algebra $\cA=\cA(V)$ which is freely generated by the
elements $\phi(f)$, $f\in V$, satisfying
\begin{enumerate}
\item $f\mapsto\phi(f)$ is $\R$-linear,
\item $\phi(f)^{*}=\phi(f)$, for all $f\in V$,
\item $\phi(f)\phi(g)-\phi(g)\phi(f)=\im\sigma(f,g)\bOne$ for all $f,g\in V$.
\end{enumerate}
Another version of CCR algebra is defined by the Weyl form of the
CCR \cite[Sec.5.2.2.2]{BR97}. Let $\fA$ be a $C^{*}$-algebra generated
by nonzero elements $W(f),$ $f\in V$, satisfying
\begin{enumerate}
\item $W(-f)=W(f)^{*},$
\item $W(f)W(g)=e^{-\im\sigma(f,g)/2}W(f+g)$ for all $f,g\in V.$
\end{enumerate}
It is shown that such $\fA$ is unique up to {*}-isomorphism. We write
$\fA=\fA(V)$ and call it the \termi{CCR $C^*$-algebra} over $V$.

Next we recall the algebraic formalism of free fields in curved spacetime,
see \cite{BGP2007,BF2009,HW2015,KM2015,BD2015}.

Let $(M,g)$ be a globally hyperbolic spacetime, where $g$ is the
metric with signature ${+}{-}{-}{-}$. The \termi{d'Alembert operator}
on $M$ is defined by $\SYM{\Box}{\Box}:=g^{ab}\nabla_{a}\nabla_{b}$.
Then for $m\ge0$, $\SYM{P_{m}}{Pm}:=\Box+m^{2}$ is called the \termi{Klein--Gordon (KG) operator}.
Let $\SYM{\GreenOp^{+}}{C+},\SYM{\GreenOp^{-}}{E-}:C_{0}^{\infty}(M,\R)\to C^{\infty}(M,\R)$
be the retarded and advanced Green operators of the KG operator, respectively
(e.g., \cite[Sec.2.1]{HW2015}). Alternatively, $\GreenOp^{\pm}$
can be viewed as distributional kernels on $M\times M$, which satisfy
\[
P_{m}\GreenOp^{\pm}(x,y)=\delta(x,y).
\]
{} Let $\GreenOp:=\GreenOp^{+}-\GreenOp^{-}$, which is sometimes called
the commutator operator. For any $f\in C_{0}^{\infty}(M)$, 

For $f_{1},f_{2}\in C_{0}^{\infty}(M)$, let $\GreenOp(f_{1},f_{2}):=(f_{1},\GreenOp(f_{2}))$.
Formally,
\[
\GreenOp(f_{1},f_{2})=\int_{M}\int_{M}f_{1}(x)f_{2}(y)\GreenOp(x,y)\d x\d y.
\]
The CCR of the Klein\textendash Gordon field operators $\{\phi(f)|f\in C_{0}^{\infty}(M)\}$
\cite{KM2015,HW2015} is given by
\begin{equation}
[\phi(f_{1}),\phi(f_{2})]=\im\GreenOp(f_{1},f_{2})\bOne,\qquad f_{1},f_{2}\in C_{0}^{\infty}(M).\label{eq:CCR-KG}
\end{equation}
However, $(C_{0}^{\infty}(M),\bE(\cdot,\cdot))$ is not a symplectic
vector space, since $\bE$ is degenerate. The quotient space $V:=C_{0}^{\infty}(M)/\cN$,
$\SYM{\cN}N:=\Ker(\GreenOp)$, becomes a symplectic vector space,
so that
\[
\phi(f+\cN):=\phi(f),\qquad\GreenOp(f+\cN,g+\cN):=\GreenOp(f,g)
\]
is well-defined for $f,g\in C_{0}^{\infty}(M)$. In fact, we have
the following (\cite[Sec.3.6, Proposition 8]{BF2009}, see also \cite[p.129]{BGP2007},
\cite[Sec.3.1.1]{BD2015}).
\begin{prop}
Let $C_{{\rm sc}}^{\infty}(M)$ denote the space of spacelike-compactly
supported smooth functions on $M$. Then the sequence
\[
0\longrightarrow C_{0}^{\infty}(M)\stackrel{P_{m}}{\longrightarrow}C_{0}^{\infty}(M)\stackrel{\GreenOp}{\longrightarrow}C_{{\rm sc}}^{\infty}(M)\stackrel{P_{m}}{\longrightarrow}C_{{\rm sc}}^{\infty}(M)
\]
is exact and complex (i.e., the composition of any two successive
maps is zero).
\end{prop}

Especially, we have 
\[
\cN=\Ker(\GreenOp)=P_{m}C_{0}^{\infty}(M)=\{P_{m}f|f\in C_{0}^{\infty}(M)\},
\]
and $P_{m}\circ\GreenOp=0$, which means that for any $f\in C_{0}^{\infty}(M)$,
$\GreenOp(f)\in C_{{\rm sc}}^{\infty}(M)$ is a solution of the KG
equation $P_{m}\varphi=0$. Thus the space $V:=C_{0}^{\infty}(M)/\Ker(\GreenOp)$
is canonically identified with the space of the solutions of the KG
equation in $C_{{\rm sc}}^{\infty}(M)$.

The KG field on $M$ is defined to be the abstract CCR algebra generated
by $\{\phi(f)|f\in V\}$ \cite{HW2015,KM2015,BD2015}, or the CCR
$C^{*}$-algebra over $V$ \cite{BGP2007,BF2009}.

Thus, a KG field in a curved spacetime is defined essentially by only
one nontrivial law, i.e., the CCR (\ref{eq:CCR-KG}). Precisely, we
implicitly assumed another law
\begin{equation}
\forall f\in\cN,\ \phi(f)=0,\label{eq:phi(f)=00003D0}
\end{equation}
which is equivalent to $\forall f\in C_{0}^{\infty}(M),\ \phi(P_{m}f)=0,$
formally written as the KG equation $P_{m}\phi(x)=0$. However, the
CCR (\ref{eq:CCR-KG}) implies that if $f\in\cN$, $\phi(f)$ becomes
``classical'' in that it commutes with $\phi(g)$ for all $g\in C_{0}^{\infty}(M)$.
Hence the law (\ref{eq:phi(f)=00003D0}) seems to have little additional
information; in any case $\phi(f)$ seems ignorable as a quantum observable
when $f\in\cN$.

\section{Some empirical laws for the CCR}

\label{sec:Some-empirical-laws}

The notions of the abstract CCR algebras given in Section \ref{sec:Definitions-of-free}
seem not very suitable for describing empirical laws directly, mainly
because they do not contain the spectral projections of the ``field
operators'' $\phi(f)$. Instead we will consider the von Neumann
algebras of operators on a Hilbert space $\cH$ concerning the CCR,
while we will actually treat only the factors of type ${\rm I}$ (mainly
${\rm I}_{\infty}$) in this paper. Alternatively we could consider
the abstract $C^{*}$-algebras which have plenty of projections, such
as abstract $W^{*}$-algebras, or more generally $AW^{*}$-algebras
\cite{Bla2006,Ber2011}.

\subsection{Prior conditional probability}

Let $\bProj$ be a set of orthogonal projections on a separable Hilbert
space $\cH$. Let $\cA$ denote the von Neumann algebra generated
by $\bProj$, i.e., $\cA:=\bProj''$, and assume that $\cA$ is a
factor of type ${\rm I}$, equivalently, that $\cA$ is isomorphic
to $B(\cK)$ (the algebra of all bounded operators on $\cK$) for
some Hilbert space $\cK$. %
{} Then $\cA$ has the canonical faithful normal semifinite trace $\Tr_{\cA}$,
corresponding to the usual trace on $B(\cK)$. Of course, when $\bProj''=B(\cH)$,
we have $\cH=\cK$ and $\Tr_{\cA}=\Tr$; When $\bProj''\neq B(\cH)$
but $\bProj$%
{} is fixed, we may redefine $\cH$ by $\cH:=\cK$, and $\Tr_{\cA}$
may be simply written as $\Tr$. In Sections \ref{sec:Some-empirical-laws}\textendash \ref{sec:Finite-dimensional-case},
we consider the cases where $\bProj$ is fixed, and we often make
such an assumption implicitly. However, for a (fixed) free scalar
field, the set $\bProj$ will not be fixed in Section \ref{sec:Free-scalar-field}.

Let $\rho\in\cA$ be a nonzero positive operator such that $\Tr_{V}\rho<\infty$,
called an \termi{unnormalized density operator} on $\cA$; the normalization
$\ul{\rho}:=\rho/\Tr_{V}(\rho)$ is called a \termi{density operator}
on $\cA$. Note that $\ul{\rho}$ (and $\rho$) may not be a trace-class
operator on $\cH$, but if we see $\ul{\rho}$ as an operator on $\cK$,
it is of trace-class, and is a density operator on $\cK$ in the usual
sense.%
{} 

Let $1<k<n$ and $E_{1},...,E_{n}\in\bProj$. Consider the conditional
probability
\begin{equation}
\Prob_{\rho}(B|A):=\frac{\Tr_{\cA}(AB)^{*}\rho AB}{\Tr_{\cA}A^{*}\rho A},\qquad A:=E_{1}\cdots E_{k},\ B:=E_{k+1}\cdots E_{n}.\label{eq:def:Prho(B|A)}
\end{equation}
If each $E_{j}$ is interpreted as a yes-no type measurement, $A=E_{1}\cdots E_{k}$
is interpreted as a successive measurement of $E_{1},...,E_{k}$;
$A$ is also a yes-no type measurement such that the output of $A$
is ``yes'' if and only if the output of $E_{j}$ is ``yes'' for
all $j=1,...,k$. The term ``the probability of $A$'' refers to
the probability that the output of $A$ is ``yes''. Here, we adopt
the convention of operational interpretation of chronological order;
the composition $EF$ of two operations $E$ and $F$ is interpreted
as the operation $F$ \emph{after} $E$. Then $\Prob_{\rho}(B|A)$
is interpreted as follows: Consider the successive measurement $AB=E_{1}\cdots E_{n}$;
$\Prob_{\rho}(B|A)$ is the probability of the yes-no type measurement
$B$, on the condition that the output of $A$ is ``yes'', when
the initial state is $\rho$.

If $\cK$ is finite-dimensional, we have $\Tr_{\cA}\bOne=\dim\cK<\infty$
for the unit $\bOne\in\cA$, and so $\Prob_{\rho}(B|A)$ is always
defined for $\rho=\bOne$. Notice that when $0<\Tr_{\cA}A^{*}A<\infty$,
we can naturally generalize (\ref{eq:def:Prho(B|A)}) for $\rho=\bOne$
even if $\dim\cK=\infty$, where $\Tr_{\cA}\bOne=+\infty$. In this
case, we call $\Prob_{\bOne}(B|A)$ a \termi{prior conditional probability}.
This can be interpreted as the probability of $B$, on the condition
that the output of $A$ is ``yes'', \emph{without no information
on the initial state. }This is expected to be suitable for the situations
where there is no ``distinguished'' state, and where we are not
sure what are the ``preparable'' states, such as QFTCS.

Next we consider the prior conditional probabilities concerning the
canonical commutation relation (CCR). Let $Q$ and $P$ be self-adjoint
operators on a Hilbert space $\cH$, which satisfy the CCR $[Q,P]=\im\hbar\bOne$
in the usual sense. Precisely, this means the Weyl relation
\[
e^{\im sQ}e^{\im tP}=e^{-\im\hbar st}e^{\im tP}e^{\im sQ},\qquad\text{for all }s,t\in\R.
\]
Let $E_{Q}(\cdot)$ and $E_{P}(\cdot)$ denote the spectral projection-valued
measures of $Q$ and $P$, respectively, and
\[
\bProj:=\{E_{Q}(S)|S\in\Borel(\R)\}\cup\{E_{P}(T)|T\in\Borel(\R)\}.
\]
Let $\cA:=\bProj''$, or equivalently $\cA:=\{e^{\im sQ},e^{\im tP}|s,t\in\R\}''$.
Then we find that $\cA$ is a factor of type ${\rm I}_{\infty}$,
by \cite[Theorem 2.4.26]{BR87} and the Stone\textendash von Neumann
uniqueness theorem (e.g., \cite[Corollary 5.2.15]{BR97}).%
{} If the CCR is irreducibly represented on $\cH$, we have $\cA=B(\cH)$.

Notice the following fact:
\begin{lem}
\label{lem:bdd-trace-class-1}Let $S$ and $T$ be bounded Borel sets
of $\R$, and $A:=E_{Q}(S)E_{P}(T).$ Then
\[
\Tr_{\cA}A=\Tr_{\cA}A^{*}A=\frac{\Leb^{1}(S)\Leb^{1}(T)}{2\pi\hbar}<\infty,
\]
where $\Leb^{1}$ is the standard Lebesgue measure on $\R^{1}$.
\end{lem}

Let $S_{j},T_{j}$ ($j=1,2,...$) be Borel subsets of $\R$, and assume
that $S_{1}$ and $T_{1}$ are bounded. Let $1\le k<n$ and
\begin{equation}
A:=E_{Q}(S_{1})E_{P}(T_{1})\cdots E_{Q}(S_{k})E_{P}(T_{k}),\qquad B:=E_{Q}(S_{k+1})E_{P}(T_{k+1})\cdots E_{Q}(S_{n})E_{P}(T_{n}).\label{eq:A=00003DEQEP}
\end{equation}
Then we have $\Tr_{\cA}A^{*}A<\infty$, and hence the prior conditional
probability $\Prob_{\bOne}(B|A)$ is defined when $\Tr_{\cA}A^{*}A>0$,
i.e., when $A\neq0$. This is equivalent to that all the sets $S_{j},T_{j}$
($j=1,...k$) are non-null in sense of the Lebesgue measure on $\R$.

\paragraph{Prior conditional probability as an empirical law.}

It is fairly reasonable to say that the explicit formula (see Proposition
\ref{prop:PI-explicit} below) to express the prior conditional probability
$\Prob_{\bOne}(B|A)$ with (\ref{eq:A=00003DEQEP}) is an empirical
law concerning the CCR. However, even if we do not know the initial
state at all, a given experimental apparatus may prepare a ``preferred''
initial state. Hence the empirical/experimental meaning of the the
prior conditional probability $\Prob_{\bOne}(B|A)$ can be unclear
in some situations, and so someone might disagree to call it an empirical
law. However, we can consider a notion of ``prior conditional probability''
in a stronger sense; Let $S_{1}$ be the open interval $(a-\epsilon,a+\epsilon)$
where $a\in\R$, $\epsilon>0$. Roughly speaking, we find that if
$p:=\lim_{\epsilon\to+0}\Prob_{\rho}(B|A)$ converges, then it is
independent of the initial state $\rho$. Although $\lim_{\epsilon\to+0}E_{Q}(S_{1})$
does not converge in norm (and it does converge in strong operator
topology but is trivial), we formally write $E_{Q}(a):=\lim_{\epsilon\to+0}E_{Q}(S_{1})$,
and $p$ by $\Prob(B|A_{a})$ where $A_{a}:=E_{Q}(a)E_{P}(S_{2})$.
We write $\Prob(B|A_{a})$ also as $\Prob_{\bOne}(B|A_{a})$ since
we have $\Prob(B|A_{a})=\lim_{\epsilon\to+0}\Prob_{\bOne}(B|A)$. 

Because the limit conditional probability $\Prob_{\bOne}(B|A_{a})$
is independent of a given initial state in a stronger sense than $\Prob_{\bOne}(B|A)$,
its empirical meaning is clearer. Thus the explicit formula to give
$\Prob_{\bOne}(B|A_{a})$ will undoubtedly be an empirical law, which
has a predictive power. Furthermore, $E_{Q}(S_{1})$ for $\epsilon\approx0$
will be easily realized as an experimental apparatus, e.g., a very
narrow slit.

\subsection{Explicit formula for the prior conditional probability of the CCR}

The explicit formula to calculate $\Prob_{\bOne}(B|A)$ with (\ref{eq:A=00003DEQEP})
is given as follows. First we prepare some notational conventions.
Let $x_{k},x_{k}',p_{k},p_{k}'$ ($k=1,...,n$) be $\R$-valued variables,
but assume $x_{1}'$ and $p_{n}'$ to be the aliases of $x_{1}$ and
$p_{n}$, respectively. Thus, the number of independent variables
is $4n-2$. For example, we interpret the double integral $\int\d x_{1}\int\d x_{1}'\ f(x_{1},x_{1}')$
{} as the integral $\int\d x_{1}\int\d x_{1}'\ \delta(x_{1}-x_{1}')f(x_{1},x_{1}')$,
which equals to the single integral $\int\d x_{1}\ f(x_{1},x_{1})$.

Let
\[
\vec{Z}\equiv(z_{1},...,z_{2n}):=(x_{1},p_{1},...,x_{n},p_{n}),
\]
\[
\vec{Z}'\equiv(z_{1}',...,z_{2n}'):=(x_{1}',p_{1}',...,x_{n}',p_{n}'),
\]
and set
\[
\int_{S_{1}\times T_{1}\times\cdots\times S_{n}\times T_{n}}\d\vec{Z}:=\int_{S_{1}}\d x_{1}\int_{T_{1}}\d p_{1}\cdots\int_{S_{n}}\d x_{n}\int_{T_{n}}\d p_{n}.
\]
Let
\[
\d_{\hbar}x_{k}:=\left(2\pi\hbar\right)^{-1/2}\d x_{k},\qquad\d_{\hbar}p_{k}:=\left(2\pi\hbar\right)^{-1/2}\d p_{k}.
\]
The notations such as $\d_{\hbar}\vec{Z}:=\left(2\pi\hbar\right)^{-n}\d\vec{Z}$
are derived from the above.

Define the function $\sfS:\R^{2n}\to\R$ by
\begin{equation}
\sfS(\vec{Z}):=\sum_{k=1}^{2n-1}\left(-1\right)^{k+1}z_{k}z_{k+1}.\label{eq:def:S(Z)}
\end{equation}

For each $n\in\N$ and $\cZ:=S_{1}\times T_{1}\times\cdots\times S_{n}\times T_{n}$
define $\sfQ(\cZ)\in\R$ by
\begin{align}
\sfQ(\cZ) & :=\int_{\cZ}\d_{\hbar}\vec{Z}\int_{\cZ}\d_{\hbar}\vec{Z}'\ e^{\im(\sfS(\vec{Z})-\sfS(\vec{Z}'))/\hbar}.\label{eq:Q(c)=00003Dint}\\
 & =\int_{\cZ}\d_{\hbar}\vec{Z}\int_{\cZ}\d_{\hbar}\vec{Z}'\ \cos\left[\frac{\im}{\hbar}(\sfS(\vec{Z})-\sfS(\vec{Z}'))\right]
\end{align}

Let $1<k<n$. Let
\[
\cX:=S_{1}\times T_{1}\times\cdots\times S_{k}\times T_{k},\qquad\cY:=S_{k+1}\times T_{k+1}\times\cdots\times S_{n}\times T_{n},
\]
so that $\cX\times\cY=S_{1}\times T_{1}\times\cdots\times S_{n}\times T_{n}=\cZ$.
Let
\begin{align*}
A & \equiv A(\cX):=E_{Q}(S_{1})E_{P}(T_{1})\cdots E_{Q}(S_{k})E_{P}(T_{k}),\\
B & \equiv B(\cY):=E_{Q}(S_{k+1})E_{P}(T_{k+1})\cdots E_{Q}(S_{n})E_{P}(T_{n}).
\end{align*}
Let
\[
\Prob_{\bOne}(\cY|\cX):=\Prob_{\bOne}(B(\cY)|A(\cX)).
\]

\begin{prop}
\label{prop:PI-explicit}Assume that $S_{1}$ and $T_{1}$ are bounded.
Then we have the following explicit formula for the prior conditional
probability:
\begin{equation}
\Prob_{\bOne}(\cY|\cX)=\frac{\sfQ(\cX\times\cY)}{\sfQ(\cX)},\label{eq:PI-explicit}
\end{equation}
if the denominator is not zero.
\end{prop}

A \emph{non-rigorous} derivation of Proposition \ref{prop:PI-explicit}
is straightforward as follows. It suffices to show it when all $S_{j}$'s
and $T_{j}$'s are bounded, where the convergence of the integral
in (\ref{eq:Q(c)=00003Dint}) is clear. We use Dirac's bra-ket notations
$|x\rangle$ and $|p\rangle$ (e.g., \cite[Sec.1.6]{Sak2011}). They
are characterized by the relations
\begin{equation}
\langle x|x'\rangle=\delta(x-x'),\qquad\langle p|p'\rangle=\delta(p-p'),\qquad\langle x|p\rangle=\frac{e^{\im px/\hbar}}{\sqrt{2\pi\hbar}}.\label{eq:braket-xp}
\end{equation}
Then we can substitute
\[
E_{Q}(S_{j})=\int_{S_{j}}\ket x\bra x\d x,\qquad E_{P}(T_{j}):=\int_{T_{j}}\ket p\bra p\d p
\]
into the rhs of (\ref{eq:def:Prho(B|A)}). Then the calculation of
$\Prob_{I}(B|A)$ is straightforward, although it is somewhat cumbersome
due to the large number of variables.

To make this argument rigorous, define $K_{S_{1},T_{1}}\in L^{2}(\R^{2})$
as follows:
\begin{equation}
K_{S_{1},T_{1}}(x_{1},x_{2}):=\frac{1}{2\pi\hbar}\chi_{S_{1}}(x_{1})\int_{T_{1}}\d p\ e^{\im p(x_{1}-x_{2})/\hbar},\qquad x_{1},x_{2}\in\R,\label{eq:Kab=00003D}
\end{equation}
where $\chi_{S_{1}}:\R\to\{0,1\}$ is the indicator function of $S_{1}\subset\R$.
Note that $K_{S_{1},T_{1}}$ is expressed by the (non-rigorous) bra-ket
notation as follows: 
\[
K_{S_{1},T_{1}}(x_{1},x_{2})=\bra{x_{1}}E_{Q}(S_{1})E_{P}(T_{1})\ket{x_{2}}.
\]
Thus we can define the Hilbert-Schmidt (in fact, trace-class) operator
$\hat{K}_{S_{1},T_{1}}$ on $L^{2}(\R)$ by
\[
\left(\hat{K}_{S_{1},T_{1}}f\right)(x_{1}):=\int\d x_{2}\ K_{S_{1},T_{1}}(x_{1},x_{2})f(x_{2}),\qquad f\in L^{2}(\R).
\]
$\hat{K}_{S_{1},T_{1}}$ is an integral-operator realization of the
operator $E_{Q}(S_{1})E_{P}(T_{1})$ on $\cH$. If we set $\cH=L^{2}(\R)$,
then $E_{Q}(S_{1})E_{P}(T_{1})$ can be identified with $\hat{K}_{S_{1},T_{1}}$.
This justifies the non-rigorous calculations with the bra-ket notations.

\subsection{Covariant reformulation}

\label{subsec:Covariant-reformulation}

The CCR is invariant under the translation $Q\mapsto Q'=Q+q_{1}\bOne$,
$P\mapsto P'=P+p_{1}\bOne$ ($q_{1},p_{1}\in\R$), and the symplectic
transformation
\[
\begin{pmatrix}Q\\
P
\end{pmatrix}\mapsto\begin{pmatrix}Q'\\
P'
\end{pmatrix}=\begin{pmatrix}a & b\\
c & d
\end{pmatrix}\begin{pmatrix}Q\\
P
\end{pmatrix},\qquad\begin{pmatrix}a & b\\
c & d
\end{pmatrix}:\text{ symplectic matrix}.
\]
We rewrite the Proposition \ref{prop:PI-explicit} so that the empirical
laws are presented in a \emph{manifestly} covariant form with respect
to the translations and the symplectic transformations. We adopt a
coordinate-free formalism where the theory becomes automatically covariant
because we can define only the covariant notions there.

Let $V$ be a real linear space. Recall that $V$ becomes an affine
space $\cV$ if we ``forget'' the origin $0$ of $V$; If we fix
an arbitrary point $o\in\cV$, then the set of vectors $V_{o}:=\{\overrightarrow{oa}|a\in\cV\}$
has the structure of real linear space isomorphic to $V$. In other
words, the additive group $V$ is the translation group on $\cV$.
For $v\in V$ and $a\in\cV$, the translation $a'\in\cV$ of $a$
w.r.t.~$v$ is written as $a'=a+v$. If $\cV$ is seen as a smooth
manifold, $V_{o}$ is naturally identified with the tangent space
$T_{o}\cV$. If $\sigma$ is a symplectic form on $V$, $(\cV,\sigma)$
is called an symplectic affine space.

Let $(\cV,\sigma)$ be a two-dimensional symplectic affine space,
which is interpreted as a classical-mechanical phase space. In classical
mechanics, for every Borel set $S\subset\cV$, there exists a yes-no
type measurement that the state of the system is exactly in $S$.
On the other hand, in quantum mechanics, such measurement exists only
if there exists $v\in V\setminus\{0\}$ such that $S$ is invariant
under the translations $\R v$.

Let $u,v\in V$ be linearly independent. For $j=1,2,...$, let $S_{j}$
(resp.~$T_{j}$) be Borel subsets of $\cV$ which is invariant under
$\R v$ (resp.~$\R u$). Assume $S_{j}/\R v$ and $T_{j}/\R u$ are
bounded.

The 2-form $\sigma$ on $\cV$ is a volume form, and so its absolute
value defines a Lebesgue measure $\d\xi$ on $\cV$. If $\vec{\xi}$
is a $\cV^{n}$-valued variable, we set
\[
\d_{\hbar}\vec{\xi}:=\left(2\pi\hbar\right)^{-n}\d\vec{\xi}.
\]

Let $n\in\N$, $\xi_{j}\in\cV$ ($j=0,...,n-1$), and define the lines
$\ell_{1},...,\ell_{2n}$ in $\cV$ by
\[
\ell_{2k+1}:=\xi_{k}+\R v_{1},\qquad\ell_{2k+2}:=\xi_{k}+\R v_{2},\qquad k=1,..,n.
\]

Let $\eta_{k}$ ($k=0,...,n-1$) be the intersection point of the
lines $\ell_{2k+2}$ and $\ell_{2k+3}$, $k=0,...,n-1$ ($\ell_{2n+1}:=\ell_{1}$).
Note that $\xi_{k}$ is the intersection point of $\ell_{2k+1}$ and
$\ell_{2k+2}$. Then we have a polygon $\xi_{0}\eta_{0}\cdots\xi_{n-1}\eta_{n-1}$
with the signed area $\bfS=\bfS(\xi_{0},...,\xi_{n-1})=\bfS(\xi_{0},\vec{\xi})$,
$\vec{\xi}:=(\xi_{1},...,\xi_{n-1})$, determined by the 2-form $\sigma$.
Precisely, $\bfS$ is defined as follows: let $c$ be the piecewise
linear closed path $\xi_{0}\to\eta_{0}\to\cdots\to\xi_{n-1}\to\eta_{n-1}\to\xi_{0}$
on $\cV$. Let $\theta$ be a 1-form such that $\d\theta=\sigma$.
Then set $\bfS:=\int_{c}\theta$, which is independent of $\theta$.

For each $n\in\N$ and $\cZ:=S_{1}\times T_{1}\times\cdots\times S_{n}\times T_{n}$,
define $\sfQ(\cZ)=\sfQ_{{\rm cov}}(\cZ)\in\R$, by
\[
\sfQ_{{\rm cov}}(\cZ):=\int_{S_{n}\cap T_{1}}\d_{\hbar}\xi_{0}\int_{\cS}\d_{\hbar}\vec{\xi}\int_{\cS}\d_{\hbar}\vec{\xi}'\ e^{\im\left(\bfS(\xi_{0},\vec{\xi})-\bfS(\xi_{0},\vec{\xi}')\right)/\hbar}
\]
where
\[
\cS:=\left(S_{1}\cap T_{2}\right)\times\cdots\times\left(S_{n-1}\cap T_{n}\right)
\]
Let $1<k<n$, and
\[
\cX:=S_{1}\times T_{1}\times\cdots\times S_{k}\times T_{k},\qquad\cY:=S_{k+1}\times T_{k+1}\times\cdots\times S_{n}\times T_{n},
\]
then the prior conditional probability is given by 
\begin{equation}
\Prob_{\bOne}(\cY|\cX)=\frac{\sfQ_{{\rm cov}}(\cX\times\cY)}{\sfQ_{{\rm cov}}(\cX)},\label{eq:PI-explicit-cov}
\end{equation}
which is of the same form as (\ref{eq:PI-explicit}).

We check that Proposition \ref{prop:PI-explicit} follows from (\ref{eq:PI-explicit-cov}).
Let
\[
V=\cV=\R^{2}=\{(x,p)|p,x\in\R\},\quad\sigma=\d x\wedge\d p,\quad u=(1,0),\quad v=(0,1),\quad\xi_{j}=(x_{j+1},p_{j+1}),
\]
and 
\[
\vec{Z}\equiv(z_{1},...,z_{2n}):=(x_{1},p_{1},...,x_{n},p_{n}),
\]
\[
\vec{Z}'\equiv(z_{1}',...,z_{2n}'):=(x_{1}',p_{1}',...,x_{n}',p_{n}'),
\]
with $x_{1}'=x_{1},p_{n}'=p_{n}$ as before. We find
\begin{align*}
\bfS(\xi_{0},\vec{\xi}) & =\sum_{k=1}^{n}x_{k+1}(p_{k+1}-p_{k}),\qquad x_{n+1}:=x_{1},\ p_{n+1}:=p_{1}\\
 & =\sum_{k=1}^{n}z_{2k+1}(z_{2k+2}-z_{2k}),\qquad z_{2n+1}:=z_{1}\\
 & =\sum_{j=1}^{2n}\left(-1\right)^{j+1}z_{j}z_{j+1}\\
 & =\sfS(\vec{Z})-x_{1}p_{n},
\end{align*}
where $\sfS(\vec{Z})$ is defined by (\ref{eq:def:S(Z)}).%
{} Thus we find
\[
\bfS(\xi_{0},\vec{\xi})-\bfS(\xi_{0}',\vec{\xi}')=\sfS(\vec{Z})-\sfS(\vec{Z}').
\]
Hence $\sfQ_{{\rm cov}}(\cZ)$ equals to $\sfQ(\cZ)$ defined by (\ref{eq:Q(c)=00003Dint}),
and so Proposition \ref{prop:PI-explicit} follows from (\ref{eq:PI-explicit-cov}).

Conversely, (\ref{eq:PI-explicit-cov}) follows from the proof of
Proposition \ref{prop:PI-explicit} in the cases where $\{Q,P\}$
is irreducible, because the translations and symplectic transformations
of the CCR are unitarily represented.

\section{Harmonic oscillator as a prototype of QFT}

\label{sec:Harmonic-oscillator}

In this section we consider a harmonic oscillator, viewed as a ``free
scalar field in $(0+1)$-dimensional spacetime'' with ``mass''
$\mu\ge0$. 

The classical Lagrangian density is given by
\[
\cL=\frac{1}{2}(\dot{x}^{2}-\mu^{2}x^{2}).
\]
The simplest case is where the field is ``massless'', $\mu=0$,
which corresponds to the a nonrelativistic free particle on a line.
In the massless case the quantized field does not have a ground state
(``vacuum''), and so the quantized field cannot be characterized
by the Wightman functions. Thus it has some significance as a prototype
of ``QFT without Wightman functions''.

\subsection{Galilei covariant free particle}

\label{subsec:Galilei-covariant-free}

Assume the usual CCR $[X,P]=\im\hbar\bOne$, which is irreducibly
represented on $\cH$. The Hamiltonian of a nonrelativistic quantum
free particle on a line is given by $H=\frac{1}{2m}P^{2}$, which
leads to
\[
X_{t}:=e^{\im tH/\hbar}Xe^{-\im tH/\hbar}=X+\frac{P}{m}t,\qquad t\in\R.
\]
Hence we have the CCR's
\begin{equation}
[X_{t_{1}},X_{t_{2}}]=\frac{-\im\hbar(t_{1}-t_{2})}{m},\qquad t_{1},t_{2}\in\R.\label{eq:CCR-massless}
\end{equation}
Roughly speaking, the CCR's of this form are seen as the \emph{definition}
of the $(0+1)$-dimensional massless free scalar field in the terminology
of \cite{BGP2007,BF2009,HW2015,KM2015,BD2015}. However, the empirical
meaning of this relation is not immediately clear.

Let $E_{t}(\cdot)$ be the spectral projection-valued measure of $X_{t}$.
Let $S_{j}$, $j=1,2,...$ be bounded Borel subsets of $\R$, and
$t_{1}<\cdots<t_{n}$. Let $E_{j}:=E_{t_{j}}(S_{j})$. For a density
matrix $\rho$ and $1<k<n$, consider the conditional probability
\begin{equation}
\Prob_{\rho}(B|A):=\frac{\Tr(AB)^{*}\rho AB}{\Tr A^{*}\rho A},\qquad A:=E_{1}\cdots E_{k},\ B:=E_{k+1}\cdots E_{n}\label{eq:def:Prho-Gali}
\end{equation}
when the denominator is not zero. Again we see that the prior conditional
probability $\Prob_{\bOne}(B|A)$ is well-defined.

The empirical laws on the system of a (classical/quantum) nonrelativistic
free particle on a line should be written in a Galilei covariant form.
Recall that the Galilean transformations on a two-dimensional spacetime
are the compositions of translations $(x,t)\mapsto(x+x_{0},t+t_{0})$
($x_{0},t_{0}\in\R$) and the ``boosts'' $(x,t)\mapsto(x+vt,t)$
($v\in\R$). Next we will present an explicit formula for $\Prob_{\bOne}(B|A)$
which is manifestly Galilei covariant, although we consider only the
case where $n=3$ and $k=2$ for simplicity and explicitness.

It is preferable that the empirical law is derived from the defining
CCR (\ref{eq:CCR-massless}) only, since in this paper we intend to
examine the empirical law of a free scalar field in curved spacetime
as an abstract CCR algebra. However, now it is easier to employ the
conventional formalism of Schr{\"o}dinger equations. Consider the
standard Schr{\"o}dinger representation of $X$ and $P$ on $L^{2}(\R)$.
Then the propagator of the Hamiltonian $H=\frac{1}{2m}P^{2}$ is given
by
\[
K(x'',x';T)=\left(\frac{m}{2\pi\im\hbar T}\right)^{1/2}\exp\left[\frac{im}{2\hbar T}(x''-x')^{2}\right].
\]
Note that this formula is not \emph{manifestly} Galilei covariant.
However, let
\begin{equation}
K_{3}\equiv K_{3}(x_{1},x_{2},x_{3};t_{1},t_{2},t_{3}):=K(x_{1},x_{2};t_{2}-t_{1})K(x_{2},x_{3};t_{3}-t_{2})K(x_{3},x_{1};t_{1}-t_{3}).\label{eq:def:K3}
\end{equation}
then we find
\[
K_{3}=\left(\frac{m^{3}}{(2\pi\im\hbar)^{3}T_{21}T_{32}T_{13}}\right)^{1/2}\exp\left[-\frac{im}{2\hbar}\frac{\left(x_{1}T_{32}+x_{2}T_{13}+x_{3}T_{21}\right)^{2}}{T_{21}T_{32}T_{13}}\right],
\]
where $T_{kl}:=t_{k}-t_{l}$. Since $\bfS:=x_{1}T_{32}+x_{2}T_{13}+x_{3}T_{21}$
is twice the signed area of the triangle $A_{1}A_{2}A_{3}$ where
$A_{k}:=(x_{k},t_{k})$ ($k=1,2,3$), $\bfS$ is Galilei invariant.
Hence $K_{3}$ is also Galilei invariant.

Consider four points $A_{1},A_{2},A_{3}$ and $A_{2}':=(x_{2}',t_{2})$.
Let
\begin{equation}
K_{4}:=K(x_{2}',x_{1},T_{12})K(x_{1},x_{2},T_{21})K(x_{2},x_{3},T_{32})K(x_{3},x_{2}',T_{23}).\label{eq:K4}
\end{equation}
Then we can check
\[
K_{4}=\frac{2\pi\hbar|T_{13}|}{m}K_{3}(x_{1},x_{2},x_{3};t_{1},t_{2},t_{3})\ \ol{K_{3}(x_{1},x_{2}',x_{3};t_{1},t_{2},t_{3})},
\]
and hence $K_{4}$ is Galilei invariant.%

{} Let $\hat{S}_{j}:=\{(x_{j},t_{j})|x_{j}\in S_{j}\}$. Define $\sfQ_{2}(\hat{S}_{1},\hat{S}_{2})\ge0$
by
\[
\sfQ_{2}(\hat{S}_{1},\hat{S}_{2}):=\int_{S_{1}}\d x_{1}\int_{S_{2}}\d x_{2}\ \left|K(x_{1},x_{2},T_{21})\right|^{2}=\frac{m}{2\pi\hbar|T_{21}|}\Leb(S_{1})\Leb(S_{2}),
\]
where $\Leb(S_{j})$ is the Lebesgue measure of $S_{j}$. Define $\sfQ_{3}(\hat{S}_{1},\hat{S}_{2},\hat{S}_{3})\in\R$
by
\[
\sfQ_{3}(\hat{S}_{1},\hat{S}_{2},\hat{S}_{3}):=\int_{S_{1}}\d x_{1}\int_{S_{2}}\d x_{2}\int_{S_{3}}\d x_{3}\int_{S_{2}}\d x_{2}'\ K_{4}.
\]
Clearly, $\sfQ_{2}$ and $\sfQ_{3}$ are Galilei invariant. We can
check
\begin{align*}
\Tr E_{2}E_{1}E_{1}E_{2} & =\Tr E_{1}E_{2}=\sfQ_{2}(\hat{S}_{1},\hat{S}_{2}),\\
\Tr E_{3}E_{2}E_{1}E_{1}E_{2}E_{3} & =\Tr E_{1}E_{2}E_{3}E_{2}=\sfQ_{3}(\hat{S}_{1},\hat{S}_{2},\hat{S}_{3}).
\end{align*}
Thus we find the following empirical law:
\begin{prop}
In the case where $k=2$, $n=3$, the prior conditional probability
$\Prob_{\bOne}(\hat{S}_{3}|\hat{S}_{1},\hat{S}_{2}):=\Prob_{\bOne}(B|A)$
defined in (\ref{eq:def:Prho-Gali}) is explicitly written Galilei
covariantly as
\[
\Prob_{\bOne}(\hat{S}_{3}|\hat{S}_{1},\hat{S}_{2})=\frac{\sfQ_{3}(\hat{S}_{1},\hat{S}_{2},\hat{S}_{3})}{\sfQ_{2}(\hat{S}_{1},\hat{S}_{2})}.
\]
\end{prop}

\subsection{Harmonic oscillator}

\label{subsec:Harmonic-oscillator}

The ``massive'' ($\mu>0$) case corresponds to a harmonic oscillator.
The Hamiltonian of a quantum harmonic oscillator in the usual unit
system is 
\[
H=\frac{P^{2}}{2m}+\frac{m\omega^{2}X^{2}}{2}.
\]
It follows that
\begin{equation}
X_{t}:=e^{\im tH/\hbar}Xe^{-\im tH/\hbar}=X\cos\omega t+\frac{P}{m\omega}\sin\omega t,\label{eq:Xt-oscill}
\end{equation}
\begin{equation}
P_{t}:=e^{\im tH/\hbar}Xe^{-\im tH/\hbar}=-m\omega X\sin\omega t+P\cos\omega t,\label{eq:Pt-oscill}
\end{equation}
 and hence we have the CCR
\begin{equation}
\left[X_{t_{1}},X_{t_{2}}\right]=\frac{\im\hbar}{m\omega}\sin\omega\left(t_{2}-t_{1}\right).\label{eq:CCR-oscill}
\end{equation}

The Hamiltonian $H$ has the well-known propagator
\[
K(x'',x';T)=K_{m,\omega}(x'',x';T)=\left(\frac{m\omega}{2\pi\im\hbar\sin\omega T}\right)^{1/2}\exp\left\{ \frac{\im m\omega}{2\hbar}\left[(x^{\prime2}+x^{\prime\prime2})\cot\omega T-2\frac{x'x''}{\sin\omega T}\right]\right\} .
\]
The prior conditional probability $\Prob_{\bOne}(\hat{S}_{3}|\hat{S}_{1},\hat{S}_{2})$
and more generally $\Prob_{\bOne}(\hat{S}_{k+1},...,\hat{S}_{n}|\hat{S}_{1},...,\hat{S}_{k})$
are expressed by this propagator similarly to the $\mu=0$ case. Note
that the Hamiltonian $H$, as well as $K_{m,\omega}$, is not invariant
under the translations $X\mapsto X+x_{0}\bOne$, $P\mapsto P+p_{0}\bOne$,
and the symplectic transformations except the rotation (\ref{eq:Xt-oscill}),
(\ref{eq:Pt-oscill}). %
However the CCR (\ref{eq:CCR-oscill}) suggests further symmetry. 

Define $K_{3}$ by (\ref{eq:def:K3}) with $K=K_{m,\omega}$. For
simplicity, let $\omega=1$, $m=1$, $\hbar=1$. Consider $\R^{2}=\{(x,p)|x,p\in\R\}$
with the symplectic form $\sigma=\d x\wedge\d p$ as a phase space.
The measurement that the value of the quantity $X_{t}=X\cos t+P\sin t$
is exactly in a Borel set $S\subset\R$ corresponds with the Borel
subset $S(t):=\{(x,p)|x\cos t+p\sin t\in S\}$ of the phase space
$\R^{2}$. Let $\ell_{x_{0},t}:=\{(x,p)|x\cos t+p\sin t=x_{0}\}$
for $x_{0}\in\R$. Consider $\ell_{k}:=\ell_{x_{k},t_{k}}$, $k=1,2,3$,
and let $A_{kl}$ be the intersection of $\ell_{k}$ and $\ell_{l}$.
The signed area $\bfS(\ell_{1},\ell_{2},\ell_{3})\equiv\bfS(A_{12},A_{23},A_{31})$
of the triangle $A_{12}A_{23}A_{31}$ is given by
\[
\bfS(A_{12},A_{23},A_{31})=-\frac{1}{2}\frac{\left(x_{1}\sin\left(t_{2}-t_{3}\right)+x_{2}\sin\left(t_{3}-t_{1}\right)+x_{3}\sin\left(t_{1}-t_{2}\right)\right)^{2}}{\sin\left(t_{2}-t_{3}\right)\sin\left(t_{3}-t_{1}\right)\sin\left(t_{1}-t_{2}\right)}.
\]
Note that $\bfS(A_{12},A_{23},A_{31})$ is invariant under translations
and symplectic transformations on the phase space. By a straightforward
calculation, we obtain
\[
K_{3}=R_{3}(t_{1},t_{2},t_{3})e^{-\im\bfS(A_{12},A_{23},A_{31})},
\]
where
\[
R_{3}(t_{1},t_{2},t_{3}):=\frac{1}{2\sqrt{2}\pi^{\frac{3}{2}}\sqrt{-\im\sin\left(t_{1}-t_{2}\right)}\,\sqrt{-\im\sin\left(t_{2}-t_{3}\right)}\,\sqrt{-\im\sin\left(t_{3}-t_{1}\right)}}.
\]
Note that $R(t_{1},t_{2},t_{3})$ is not symplectic invariant. Let
\[
\ell_{2}':=\ell_{x_{2}',t_{2}}:=\{(x,p)|x\cos t_{2}+p\sin t_{2}=x_{2}'\}
\]
Define $K_{4}$ by (\ref{eq:K4}) with $K=K_{\omega,m}$.%
{} Explicitly, we find 
\[
K_{4}=R_{4}(t_{1},t_{2},t_{3})e^{-\im(\bfS(\ell_{1},\ell_{2},\ell_{3})-\bfS(\ell_{1},\ell_{2}',\ell_{3}))},
\]
where
\[
R_{4}(t_{1},t_{2},t_{3}):=\frac{1}{4\pi^{2}\left|\sin\left(t_{1}-t_{2}\right)\sin\left(t_{2}-t_{3}\right)\right|}.
\]

Let $E_{t}(\cdot)$ be the spectral projection-valued measure of $X_{t}$,
and $E_{k}:=E_{t_{k}}(S_{k})$. We have
\begin{equation}
\Tr E_{1}E_{2}E_{3}E_{2}=\int_{S_{1}}\d x_{1}\int_{S_{2}}\d x_{2}\int_{S_{3}}\d x_{3}\int_{S_{2}}\d x_{2}'\ K_{4}\label{eq:TrE1E2E3=00003D}
\end{equation}
which is well-defined when $S_{1},S_{2}$ are bounded and $t_{k}-t_{l}\notin\Z\pi$
for $k\neq l$. Let $\hat{S}_{k}:=S_{k}(t_{k})=\{(x,p)|x\cos t_{k}+p\sin t_{k}\in S_{k}\}$,
and write $E(\hat{S}_{k}):=E_{k}$. Since $\Tr E_{1}E_{2}E_{3}E_{2}=\Tr E(\hat{S}_{1})E(\hat{S}_{2})E(\hat{S}_{3})E(\hat{S}_{2})$
is invariant under translations and symplectic transformations, the
r.h.s.~of (\ref{eq:TrE1E2E3=00003D}) can be rewritten in a covariant
form; The integral measure $\d x_{1}\d x_{2}\d x_{3}\d x_{2}'$ is
not symplectic invariant, and instead $R_{4}(t_{1},t_{2},t_{3})\d x_{1}\,\d x_{2}\,\d x_{3}\,\d x_{4}$
should be invariant. In fact, for fixed $t_{k},t_{l}$, the map $(x_{k},x_{l})\mapsto A_{kl}=(x_{kl},p_{kl})$
is linear, e.g.,
\[
x_{12}=\frac{x_{2}\sin t_{1}-x_{1}\sin t_{2}}{\sin(t_{1}-t_{2})},\qquad p_{12}=\frac{x_{1}\cos t_{2}-x_{2}\cos t_{1}}{\sin(t_{1}-t_{2})}
\]
and hence the invariant measure $\d x_{12}\d p_{12}$ is written as
\[
\d x_{12}\d p_{12}=\frac{1}{|\sin(t_{2}-t_{1})|}\d x_{1}\d x_{2}.
\]
We write $\d A=\d x\,\d p$ for $A=(x,p)$. Observe that for each
fixed $t_{1},t_{2},t_{3}$, $K_{4}$ is determined by $A_{12}:=\ell_{1}\cap\ell_{2}$
and $A_{23}':=\ell_{2}'\cap\ell_{3}$. Thus we write
\[
K_{4}'(A_{12},A_{23}'):=e^{-\im(\bfS(\ell_{1},\ell_{2},\ell_{3})-\bfS(\ell_{1},\ell_{2}',\ell_{3}))}
\]
and we have the manifestly covariant expression:
\[
\sfQ(\hat{S}_{1},\hat{S}_{2},\hat{S}_{3}):=\Tr E(\hat{S}_{1})E(\hat{S}_{2})E(\hat{S}_{3})E(\hat{S}_{2})=(2\pi)^{-2}\int_{\hat{S}_{1}\cap\hat{S}_{2}}\d A_{12}\int_{\hat{S}_{2}\cap\hat{S}_{3}}\d A_{23}'\ K_{4}'(A_{12},A_{23}')
\]
For $t_{1},t_{2}$ fixed, $|K_{m,\omega}|^{2}$ is the constant %
$\left(2\pi|\sin(t_{2}-t_{1})|\right)^{-1}$, hence
\[
\sfQ(\hat{S}_{1},\hat{S}_{2}):=\Tr E(\hat{S}_{1})E(\hat{S}_{2})=\int_{S_{1}}\d x_{1}\int_{S_{2}}\d x_{2}\frac{1}{2\pi|\sin(t_{2}-t_{1})|}=\frac{1}{2\pi}\Leb(\hat{S}_{1}\cap\hat{S}_{2}),
\]
where $\Leb(\cdot)$ is the Lebesgue measure on $\R^{2}$ determined
by $\d x\,\d p$, which is of course invariant. Thus we arrived at
an empirical law as a prior conditional probability in a covariant
form:
\begin{equation}
\Prob_{\bOne}(\hat{S}_{3}|\hat{S}_{1},\hat{S}_{2})=\frac{\sfQ(\hat{S}_{1},\hat{S}_{2},\hat{S}_{3})}{\sfQ(\hat{S}_{1},\hat{S}_{2})}.\label{eq:P1-harmonic}
\end{equation}

\section{Finite-dimensional CCR}

\label{sec:Finite-dimensional-case}

Let $(V,\sigma)$ be a finite-dimensional symplectic vector space,
and $(\cV,\sigma)$ be the corresponding symplectic affine space,
whose translation group is the additive group $V$. Here we view $\cV$
and $V$ to be identical as sets, but they have different structures.
Note that we can define only the translation-invariant notions on
$\cV$. Let $\cH$ be a Hilbert space. For $f\in V$, let $\phi(f)$
be a self-adjoint operator such that $f\mapsto\phi(f)$ is linear,
and satisfy the CCR $[\phi(f),\phi(g)]=\im\sigma(f,g)\bOne$ in the
usual sense. We assume that this representation of the CCR is irreducible.
In this case the system $\{\phi(f)|f\in V\}$ is unique up to unitary
equivalence.
\begin{defn}
For $F\subset V$, let
\[
F^{\sperp}:=\{f\in V|\forall g\in F,\ \sigma(f,g)=0\}.
\]

Let $W$ be a linear subspace of $V$.
\begin{itemize}
\item $W$ is \termi{isotropic} if $W\subset W^{\sperp}$; Equivalently
if $\sigma(f,g)=0$ for all $f,g\in W$.
\item $W$ is \termi{coisotropic} if $W\supset W^{\sperp}$; Equivalently,
if $W^{\sperp}$ is isotropic.
\item $W$ is \termi{Lagrangian} if $W=W^{\sperp}$; Equivalently, $W$
is isotropic and $\dim W=\frac{1}{2}\dim V$.
\item $W$ is \termi{symplectic} if $W\cap W^{\sperp}=\{0\}$.
\end{itemize}
\end{defn}

In the literature, $W^{\sperp}$ is often denoted by $W^{\perp}$,
$W^{\sigma}$, $W^{\circ}$, etc.

For each Borel set $S\subset\R$, and $f\in V\setminus\{0\}$, let
$E_{f}(S)$ be the spectral projection of $\phi(f)$. Let $C_{f}(S):=\{g\in V|\sigma(f,g)\in S\}$.
If $S_{1},S_{2}\subset\R$ are open sets, we see that the following
three conditions are equivalent:
\begin{enumerate}
\item $E_{f_{1}}(S_{1})=E_{f_{2}}(S_{2})$,
\item $f_{1}=rf_{2}$ and $S_{1}=rS_{2}$ for some $r\in\R$,
\item $C_{f_{1}}(S_{1})=C_{f_{2}}(S_{2})$.
\end{enumerate}
For general Borel sets, the situation is somewhat more complex: $E_{f_{1}}(S_{1})=E_{f_{2}}(S_{2})$
iff $\exists r\in\R,f_{1}=rf_{2}$, $S_{1}\sim rS_{2}$, where $X\sim Y$
means that $X$ and $Y$ are ``almost equal'', that is, the symmetric
difference $X\triangle Y$ is a Lebesgue null set. Hence the above
(1) and (3) are not equivalent precisely. Instead we can consider
the quotient set $\cB(\R):=\Borel(\R)/{\sim}$, which is a complete
Boolean algebra, and define $E_{f}(\tilde{S}):=E_{f}(S)$ for $\tilde{S}=[S]_{\sim}\in\cB(\R)$,
the equivalence class of $S\in\Borel(\R)$. In this case we should
understand the ``set'' $C_{f}(\tilde{S})$ to be an element of the
Boolean algebra $\cB(V):=\Borel(V)/{\sim}$. In the following, $E_{f}(S)$
and $C_{f}(S)$ are understood in this way. An element of $\cB(V)$
will be loosely called a ``subset'' of $V$, similarly to the usual
functional-analytic terminology: $f\in L^{p}(V)$ is called a ``function''
on $V$.

Thus the elements of the sets $\SYM{\cE_{1}}{E1}:=\{E_{f}(S)|S\in\cB(\R),f\in V\setminus\{0\}\}$
and $\SYM{\cO_{1}}{O1}:=\{C_{f}(S)|S\in\cB(\R),f\in V\setminus\{0\}\}$
are in one-to-one correspondence, and so we can write the element
of $\cE_{1}$ corresponding to $X\in\cO_{1}$ by $E(X)$, and conversely
the element of $\cC_{1}$ corresponding to $E\in\cE_{1}$ by $C(E)$.

Let $f_{1},f_{2}\in V\setminus\{0\}$ satisfy $\sigma(f_{1},f_{2})=0$,
so that $\phi(f_{1})$ and $\phi(f_{2})$ are commutative. Then the
projection $E_{f_{1}}(S_{1})E_{f_{2}}(S_{2})$ naturally corresponds
to
\[
\cO_{f_{1}}(S_{1})\cap\cO_{f_{2}}(S_{2})=\{g\in V|\sigma(f_{k},g)\in S_{k},k=1,2\}.
\]
More generally, let $E_{f_{1},f_{2}}(S)$ ($S\in\cB(\R^{2})$) be
the simultaneous spectral projection of $(\phi(f_{1}),\phi(f_{2}))$.
Then it naturally corresponds to the set
\[
\cO_{f_{1},f_{2}}(S):=\{g\in V|(\sigma(f_{1},g),\sigma(f_{2},g))\in S\}.
\]
Thus the elements of the sets
\[
\cE_{2}:=\{E_{f_{1},f_{2}}(S)|S\in\cB(\R^{2}),f_{1},f_{2}\in V\setminus\{0\},\ \sigma(f_{1},f_{2})=0\}
\]
 and
\[
\cO_{2}:=\{C_{f_{1},f_{2}}(S)|S\in\cB(\R^{2}),f_{1},f_{2}\in V\setminus\{0\},\ \sigma(f_{1},f_{2})=0\}
\]
are in one-to-one correspondence. %

More generally, let $1\le k\le\frac{1}{2}\dim V$, $S\in\cB(\R^{k})$,
and $f_{1},...,f_{k}\in V\setminus\{0\}$ be pairwise orthogonal,
i.e., $\sigma(f_{i},f_{j})=0$ for all $i,j=1,...,k$. Let
\[
\SYM{C_{f_{1},...,f_{k}}(S)}{Cf..f}:=\{g\in V|(\sigma(f_{1},g),...,\sigma(f_{k},g))\in S\}.
\]
Define $\SYM{\cO_{k}}{Ok}=\SYM{\cO_{k}(V)}{Ok(V)}\subset\cB(V)$ to
be such that $B\in\cC_{k}$ iff $B=C_{f_{1},...,f_{k}}(S)$ for some
$f_{1},...,f_{k}$ and $S$. Clearly $\cO_{k}\subset\cO_{l}$ if $k<l$.
Let
\[
\SYM{\cO}O=\SYM{\cO(V)}{O(V)}:=\cO_{n},\qquad n:=\frac{1}{2}\dim V.
\]

For a subset $X$ of $\cV$ (or $V$), define the subspace $\Inv_{V}(X)$
of $V$ by
\[
\Inv_{V}(X):=\{v\in V|X+rv=X,\ \forall r\in\R\}.
\]

\begin{prop}
Let $B\in\cB(V)$. Then $B\in\cO(V)$ if and only if $\Inv_{V}(B)$
is coisotropic.
\end{prop}

\begin{proof}
If $\dim\Inv_{\R^{k}}(S)\ge1$, there exist $f_{1}',...,f_{l}'\in\Span\{f_{1},...,f_{k}\}$
with $l<k$, and $S'\in\cB(\R^{l})$ such that $\Inv_{\R^{l}}(S')=\{0\}$
and $C_{f_{1},...,f_{k}}(S)=C_{f_{1}',...,f_{l}'}(S')$. Hence we
can assume $\Inv_{\R^{k}}(S)=\{0\}$ without loss of generality. Let
$B=C_{f_{1},...,f_{k}}(S)\in\cO(V)$ for some $S\in\cB(V)$ with $\Inv_{\R^{k}}(S)=\{0\}$,
and for some pairwise diagonal $f_{1},...,f_{k}\in V\setminus\{0\}$.
Then we see $\Inv_{V}(B)=\{f_{1},...,f_{k}\}^{\sperp}$, which is
coisotropic.

Conversely, assume that $B\in\cB(V)$ and $\Inv_{V}(B)$ is coisotropic,
equivalently, that $\Inv_{V}(B)^{\sperp}$ is isotropic. Let $f_{1},...,f_{k}\in\Inv_{V}(B)^{\sperp}$
be a basis of $\Inv_{V}(B)^{\sperp}$, so that $\Inv_{V}(B)=\{f_{1},...,f_{k}\}^{\sperp}$.
Let $S:=\{(\sigma(f_{1},v),...,\sigma(f_{k},v))|v\in B\}\subset\R^{2}$.
Then
\begin{align*}
v\in C_{f_{1},...,f_{k}}(S) & \iff\exists v'\in B,\forall j\in\{1,...,k\},\ \sigma(f_{j},v)=\sigma(f_{j},v')\\
 & \iff\exists v'\in B,\ v'-v\in\{f_{1},...,f_{k}\}^{\sperp}\\
 & \iff\exists v'\in B,\ v'-v\in\Inv_{V}(B)\\
 & \iff v\in B.
\end{align*}
Hence $B\in\cO(V)$.
\end{proof}
Thus we can define $\cO(\cV)$ as a subset of $\cB(\cV)$, not of
$\cB(V)$, with a manifest translation covariance:
\[
\SYM{\cO(\cV)}{O(V)}:=\left\{ B\in\cB(\cV)|\Inv_{V}(B)\text{ is coisotropic}\right\} .
\]
Any element of $\cO(\cV)$ is called an \termi{observable subset}
of $\cV$. The $B\in\cO(\cV)$ is called a \termi{Lagrangian observable subset}
of $\cV$ if $\Inv_{V}(B)$ is Lagrangian. For an observable set $B\in\cB(\cV)$,
let $\SYM{E(B)}{E(B)}$ denote the projection operator corresponding
to $B$.
\begin{prop}
\label{prop:finite-dim-TrB1B2}Let $\Leb_{\cV}:\cB(\cV)\to[0,\infty]$
denote the Lebesgue measure on $\cV$ normalized by the volume form
$(2\pi)^{-n}\sigma^{\wedge n}/n!$. If $B_{1}$ and $B_{2}$ are Lagrangian
observable sets of $\cV$, then
\[
\Tr E(B_{1})E(B_{2})=\Leb_{\cV}(B_{1}\cap B_{2}).
\]
Here, possibly both sides equal to $+\infty$. %
{} If $\Inv_{V}(B_{1})\cap\Inv_{V}(B_{2})=\{0\}$ and $B_{j}/\Inv_{V}(B_{j})$
is bounded for $j=1,2$ , then $\Tr E(B_{1})E(B_{2})<\infty$.
\end{prop}

Before we prove Proposition \ref{prop:finite-dim-TrB1B2}, recall
some facts on symplectic vector spaces, see e.g., \cite[Sec.1.3]{Ber2001}.
Let $\SYM{\Lag(V)}{Lag()}$ denote the set of Lagrangian subspaces
of $V$. Then the symplectic transformation group $Sp(V)$ transitively
acts on $\Lag(V)$, i.e., for any $L_{1},L_{2}\in\Lag(V)$, there
is a $\psi\in Sp(V)$ such that $\psi(L_{1})=L_{2}$. Let $L\in\Lag(V)$
and $G_{L}$ be the isotropy subgroup of $Sp(V)$ w.r.t.~$L$, i.e.,
\[
\SYM{G_{L}}{GL}=\{\psi\in Sp(V)|\psi(L)=\psi\}.
\]
Let $\Tv(L)$ be the set of $L'\in\Lag(V)$ transverse to $L'$:
\[
\SYM{\Tv(L)}{Tv()}:=\{L'\in\Lag(V)|L\cap L'=\{0\}\}.
\]
Then $G_{L}$ acts transitively on $\Tv(L)$. Thus we have
\begin{lem}
\label{lem:LagLag-SymFrame}Let $L_{1}\in\Lag(V)$ and $L_{2}\in\Tv(L_{1})$.
Then there exists a basis $e_{1},...,e_{n}$ ($n:=\frac{1}{2}\dim V$)
of $L_{1}$ and a basis $f_{1},...,f_{n}$ of $L_{2}$ such that
\[
\sigma(e_{k},e_{l})=0,\qquad\sigma(e_{k},f_{l})=\delta_{kl},\qquad\sigma(f_{k},f_{l})=0,
\]
that is, $(e_{1},...,e_{n},f_{1},...,f_{n})$ is a symplectic frame
of $V$.
\end{lem}

\begin{proof}
\emph{of Proposition} \ref{prop:finite-dim-TrB1B2}. Let $L_{j}:=\Inv_{V}(B_{j})$
for $j=1,2$, and $(e_{1},...,e_{n},f_{1},...,f_{n})$ be a symplectic
frame of $V$ given in Lemma \ref{lem:LagLag-SymFrame}. Let
\[
Q_{j}:=\phi(e_{j}),\qquad P_{j}:=\phi(f_{j})
\]
so that $[Q_{k},P_{l}]=\im\delta_{kl}\bOne$, $k,l=1,...,n$. Denote
the element $\sum_{j=1}^{n}x_{j}e_{j}+\sum_{l=1}^{n}y_{l}f_{l}$ of
$V$ (and $\cV$) by the coordinate $(x_{1},...,x_{n},y_{1},...,y_{n})\in\R^{2n}$.
Let $\vec{x}:=(x_{1},...,x_{n})$, $\vec{y}:=(y_{1},...,y_{n})$.
For $S,T\in\cB(\R^{n})$, we find
\[
C_{e_{1},...,e_{n}}(S)=\{(\vec{x},\vec{y})\in\R^{2n}=V|\vec{y}\in S\},
\]
\[
C_{f_{1},...,f_{n}}(T)=\{(\vec{x},\vec{y})\in\R^{2n}=V|-\vec{x}\in T\}.
\]
Hence $C_{e_{1},...,e_{n}}(S)\cap C_{f_{1},...,f_{n}}(T)=(-T)\times S.$
For $S\in\cB(\R^{k})$, let $E_{Q_{1},...,Q_{k}}(S)$ (resp.~$E_{P_{1},...,P_{k}}(S)$)
be the simultaneous spectral projection of $(Q_{1},...,Q_{k})$ (resp.~$(P_{1},...,P_{k})$)
w.r.t.~$S$. %
Let $Q,P$ be a self-adjoint CCR pair, $[Q,P]=\im\bOne$, irreducibly
represented on a Hilbert space $\cK$. If $S=S_{1}\times\cdots\times S_{n}$,
$T=T_{1}\times\cdots\times T_{n}$, $S_{j},T_{j}\in\cB(\R)$, $j=1,...,n$,
we see that $E_{Q_{1},...,Q_{k}}(S)$ (resp.~$E_{P_{1},...,P_{k}}(T)$)
is unitarily equivalent to $E_{Q}(S_{1})\otimes\cdots\otimes E_{Q}(S_{k})$
(resp.~$E_{P}(T_{1})\otimes\cdots\otimes E_{P}(T_{k})$). %
Let $B_{1}:=C_{e_{1},...,e_{n}}(S)$ and $B_{2}:=C_{f_{1},...,f_{n}}(T)$,
then we have 
\begin{align*}
\Tr E(B_{1})E(B_{2}) & =\Tr E_{Q_{1},...,Q_{n}}(S)E_{P_{1},...,P_{n}}(T)\\
 & =\prod_{j=1}^{n}\Tr E_{Q}(S_{j})E_{P}(T_{j})\\
 & =\prod_{j=1}^{n}\left(2\pi\right)^{-1}\Leb_{\R}(S_{j})\Leb_{\R}(T_{j})\\
 & =\left(2\pi\right)^{-n}\Leb_{\R^{2n}}(T\times S)\\
 & =\left(2\pi\right)^{-n}\Leb_{\R^{2n}}(B_{1}\cap B_{2}).
\end{align*}
Therefore, $\Tr E(B_{1})E(B_{2})=\left(2\pi\right)^{-n}\Leb_{\R^{2n}}(B_{1}\cap B_{2})$
holds for general $S,T\in\cB(\R^{n}).$
\end{proof}
Let $B_{1},...,B_{n}\in\cO(\cV)$, and assume that $B_{1}$ and $B_{2}$
are nonzero Lagrangian observable sets, $\Inv_{V}(B_{1})\cap\Inv_{V}(B_{2})=\{0\}$
and that $B_{j}/\Inv_{V}(B_{j})$ is bounded for $j=1,2$. (Note that
an element $B$ of the Boolean algebra $\cB(\cV)$ is nonzero iff
the Lebesgue measure of $B$ is nonzero.) Then we can define the following
prior conditional probability, invariant w.r.t.~symplectic transformations
and translations:
\begin{equation}
\Prob_{\bOne}(B_{k+1},...,B_{n}|B_{1},...,B_{k})=\frac{\sfQ(B_{1},...,B_{n})}{\sfQ(B_{1},...,B_{k})}.\label{eq:P1-finite-dim}
\end{equation}
\[
\sfQ(B_{1},...,B_{j}):=\Tr A^{*}A,\qquad A:=E(B_{1})\cdots E(B_{j}).
\]
This may be seen as a kind of empirical law. However this is too implicit
to be called an empirical law in the usual sense, since $\sfQ$ is
not explicitly defined so that one can calculate the value of the
prior conditional probability from the data of $B_{1},...,B_{n}$.

Although we will not present any general formula to give $\sfQ$ explicitly
in this paper, here we give a formula in a special case, generalizing
(\ref{eq:P1-harmonic}).

Let $A_{1},A_{2},A_{3}\in\cV$. We define the signed area $\SYM{\bfS(A_{1},A_{2},A_{3})}{S(,,)}$
of the triangle $A_{1}A_{2}A_{3}$ to be measured by the 2-form $\sigma$
on $\cV$: Let $D$ be the interior of the triangle, then $\bfS(A_{1},A_{2},A_{3}):=\int_{D}\sigma$.
The area of the polygon $A_{1}\cdots A_{k}$ is expressed by $\sum_{j=2}^{k-1}\bfS(A_{1},A_{j},A_{j+1})$,
especially $\bfS(A_{1},A_{2},A_{3})-\bfS(A_{1},A_{4},A_{3})$ for
$k=4$.

Fix Lagrangian subspaces $W_{j}$, $j=1,2,3$ such that $W_{i}\cap W_{j}=\{0\}$
for $i\neq j$. For $X_{j}\in\cV/W_{j}=\{A+W_{j}|A\in\cV\}$ ($j=1,2,3$)
and $X_{2}'\in\cV/W_{2}$, let $A_{ij}:=X_{i}\cap X_{j}$, $A_{23}':=X_{2}'\cap X_{3}$.
Consider
\[
\bfS(X_{1},X_{2},X_{3},X_{2}'):=\bfS(A_{12},A_{23},A_{31})-\bfS(A_{12},A_{23}',A_{31})
\]
Since
\[
X_{1}=A_{12}+W_{1},\ X_{2}=A_{12}+W_{2},\ X_{2}'=A_{23}'+W_{2},\ X_{3}=A_{23}'+W_{3},\ 
\]
$\bfS(X_{1},X_{2},X_{3},X_{2}')$ is determined by $A_{12}$ and $A_{23}'$.
Hence the following expression makes sense:
\[
\int_{D_{1}}\d A_{12}\int_{D_{2}}\d A_{23}'\ e^{-\im\bfS(X_{1},X_{2},X_{3},X_{2}')},
\]
where $\d A_{12}$ and $\d A_{23}'$ denote the integral w.r.t.~the
Lebesgue measure $\Leb_{\cV}$, and $D_{1},D_{2}\subset\cV$.

\begin{lem}
\label{lem:xSx}Let $W_{j}\in\Lag(V)$, $j=1,2,3$ satisfy $j\neq k\then W_{j}\cap W_{k}=\{0\}$.
Then there exists a symplectic basis $(e_{1},...,e_{n},f_{1},...,f_{n})$
such that $W_{j}$'s are expressed in the corresponding coordinates
$(\vec{x},\vec{p})=(x_{1},...,x_{n},p_{1},...,p_{n})$ as follows:
\begin{equation}
W_{1}=\left\{ (\vec{x},\vec{0})|\vec{x}\in\R^{n}\right\} ,\quad W_{2}=\left\{ (\vec{0},\vec{p})|\vec{p}\in\R^{n}\right\} ,\label{eq:L1=00003DL2=00003D}
\end{equation}
\[
W_{3}=\left\{ (\vec{x},A\vec{x})|\vec{x}\in\R^{n}\right\} 
\]
for some matrix $A\in\GL(n,\R)$ with $A^{\TT}=A$.
\end{lem}

\begin{proof}
See e.g.~\cite[Corollary 1.23]{deG2006}.
\end{proof}
\begin{lem}
Let $W_{1},W_{2}$ be as (\ref{eq:L1=00003DL2=00003D}). The subgroup
of $\Sp(V)$ which fixes $W_{1}$ and $W_{2}$ is identified with
the group of matrices
\[
\begin{pmatrix}T & 0\\
0 & T^{\TT\,-1}
\end{pmatrix},\qquad T\in\GL(n,\R).
\]
\end{lem}

\begin{proof}
See e.g.~\cite[p.23]{Ber2001}.
\end{proof}
\begin{thm}
\label{prop:Q(B1B2B3)-explicit}Let $B_{1},B_{2},B_{3}$ be Lagrangian
observable subsets of $\cV$, and $W_{i}:=\Inv_{V}(B_{i})$. $i=1,2,3$.
Assume that (1) $W_{i}\cap W_{j}=\{0\}$ for $i\neq j$, (2) $B_{j}/W_{j}$
is bounded for $j=1,2,3$, (3) there exists a symplectic basis $(e_{1},...,e_{n},f_{1},...,f_{n})$
such that the matrix $A\in\GL(n,\R)$ of Lemma \ref{lem:xSx} is diagonal.
Define $\bfS(X_{1},X_{2},X_{3},X_{2}')$ as above. Then
\[
\sfQ(B_{1},B_{2},B_{3})=\int_{B_{1}\cap B_{2}}\d A_{12}\int_{B_{2}\cap B_{3}}\d A_{23}'\ e^{-\im\bfS(X_{1},X_{2},X_{3},X_{2}')}.
\]
\end{thm}

\begin{proof}
First, recall the following. Let $(e_{1},...,e_{n},f_{1},...,f_{n})$
be a symplectic frame of $V$, so that $V$ is identified with $\{(q_{1},...,q_{n},p_{1},...,p_{n})\in\R^{2n}\}$.
Let
\[
A^{(i)}=(q_{1}^{(i)},...,q_{n}^{(i)},p_{1}^{(i)},...,p_{n}^{(i)})\in\R^{2n},\qquad i=1,2,3.
\]
Let $A_{j}^{(i)}:=\left(q_{j}^{(i)},p_{j}^{(i)}\right)\in\R^{2}$,
$j=1,...,n,$ $i=1,2,3$. The signed area of the triangle $A_{j}^{(1)}A_{j}^{(2)}A_{j}^{(3)}$
on $\R^{2}$ is given by
\[
\sfS\left(A_{j}^{(1)},A_{j}^{(2)},A_{j}^{(3)}\right)=\frac{1}{2}\sum_{i=1}^{3}\left(p_{j}^{(i+1)}q_{j}^{(i)}-p_{j}^{(i)}q_{j}^{(i+1)}\right),\qquad q_{j}^{(4)}:=q_{j}^{(1)},\ p_{j}^{(4)}:=p_{j}^{(1)}.
\]
Hence we find that the signed area of the triangle $A^{(1)}A^{(2)}A^{(3)}$
on $\R^{2n}$, measured by the 2-form $\sigma=\sum_{j=1}^{n}\d q_{j}\wedge\d p_{j}$,
is given by
\[
\sfS(A^{(1)},A^{(2)},A^{(3)})=\sum_{j=1}^{n}\sfS\left(A_{j}^{(1)},A_{j}^{(2)},A_{j}^{(3)}\right),
\]
where $\sfS\left(A_{j}^{(1)},A_{j}^{(2)},A_{j}^{(3)}\right)$ is the
signed area of the triangle $A_{j}^{(1)}A_{j}^{(2)}A_{j}^{(3)}$ on
$\R^{2}$. Thus we have the expression
\begin{equation}
\bfS(X_{1},X_{2},X_{3},X_{2}')=\sum_{j=1}^{n}\left[\sfS(A_{12,j},A_{23,j},A_{31,j})-\sfS(A_{12,j},A_{23,j}',A_{31,j})\right],\label{eq:S(XXXX)=00003Dsum}
\end{equation}
where e.g., $A_{12}=(x_{12,1},...,x_{12,n},p_{12,1},...,p_{12,n})\in\R^{2n}$
and $A_{12,j}=(x_{12,j},p_{12,j})\in\R^{2}$.

Assume that the matrix $A$ of Lemma \ref{lem:xSx} is diagonal: $A={\rm diag}(a_{1},...,a_{n})$
($a_{1},...,a_{n}\in\R\setminus\{0\}$). Let 
\[
Q_{j}:=\phi(e_{j}),\qquad P_{j}:=\phi(f_{j}),\qquad j=1,...,n,
\]
so that $[Q_{k},P_{l}]=\im\delta_{kl}\bOne$, $k,l=1,...,n$, and
\[
R_{j}:=\phi(e_{j}')=Q_{j}+a_{j}P_{j},\qquad e_{j}':=e_{j}+a_{j}f_{j}.
\]
For $S\in\cB(\R^{n})$, for let $E_{Q_{1},...,Q_{n}}(S)$ be the simultaneous
spectral projection of $(Q_{1},...,Q_{k})$ w.r.t.~$S$; similarly
for $E_{P_{1},...,P_{n}}(S)$ and $E_{R_{1},...,R_{n}}(S).$ 

Let $Q,P$ be a self-adjoint CCR pair, $[Q,P]=\im\bOne$, irreducibly
represented on a Hilbert space $\cK$.

If $S^{(i)}=S_{1}^{(i)}\times\cdots\times S_{n}^{(i)}$ $S_{j}^{(i)}\in\cB(\R)$,
$j=1,...,n$, $i=1,2,3$, we see that there exists a unitary operator
$U:\cH\to\cK^{\otimes n}$ such that
\begin{align*}
UE_{Q_{1},...,Q_{n}}(S^{(1)})U^{-1} & =E_{Q}(S_{1}^{(1)})\otimes\cdots\otimes E_{Q}(S_{n}^{(1)}),\\
UE_{P_{1},...,P_{n}}(S^{(2)})U^{-1} & =E_{P}(S_{1}^{(2)})\otimes\cdots\otimes E_{P}(S_{n}^{(2)}),\\
UE_{R_{1},...,R_{n}}(S^{(3)})U^{-1} & =E_{Q+a_{1}P}(S_{1}^{(3)})\otimes\cdots\otimes E_{Q+a_{n}P}(S_{n}^{(3)}).
\end{align*}
Let $B_{1}:=C_{e_{1},...,e_{n}}(S^{(1)})$, $B_{2}:=C_{f_{1},...,f_{n}}(S^{(2)})$
and $B_{3}:=C_{e_{1}',...,e_{n}'}(S^{(3)})$, and
\[
B_{1,j}:=S_{j}^{(1)}\times\R=\left\{ (x,p)\in\R^{2}|x\in S_{j}^{(1)}\right\} 
\]
\[
B_{2,j}:=\R\times S_{j}^{(2)}=\left\{ (x,p)\in\R^{2}|p\in S_{j}^{(2)}\right\} 
\]
\[
B_{3,j}:=\left\{ (x,p)\in\R^{2}|x+a_{k}p\in S_{j}^{(3)}\right\} 
\]
For an observable subset $B$ of the two-dimensional phase space $\R^{2}=\{(x,p)|x,p\in\R\}$,
let $E_{2}(B)$ denote the corresponding projection operator on $\cK$.
Recall that the proposition holds when $n=1$ by (\ref{eq:P1-harmonic}).
Thus we have
\begin{align*}
 & \Tr E(B_{1})E(B_{2})E(B_{3})E(B_{2})\\
 & \quad=\Tr E_{Q_{1},...,Q_{n}}(S^{(1)})E_{P_{1},...,P_{n}}(S^{(2)})E_{R_{1},...,R_{k}}(S^{(3)})E_{P_{1},...,P_{n}}(S^{(2)})\\
 & \quad=\prod_{j=1}^{n}\Tr E_{Q}(S_{j}^{(1)})E_{P}(S_{j}^{(2)})E_{Q+a_{k}P}(S_{j}^{(3)})E_{P}(S_{j}^{(2)})\\
 & \quad=\prod_{j=1}^{n}\Tr E_{2}(B_{1,j})E_{2}(B_{2,j})E_{2}(B_{3,j})E_{2}(B_{2,j})\\
 & \quad=\prod_{j=1}^{n}\int_{B_{1,j}\cap B_{2,j}}\d A_{12,j}\int_{B_{2,j}\cap B_{3,j}}\d A_{23,j}'\ e^{-\im\left[\bfS(A_{12,j},A_{23,j},A_{31,j})-\bfS(A_{12,j},A_{23,j}',A_{31,j})\right]}\\
 & \quad=\left(\prod_{j=1}^{n}\int_{B_{1,j}\cap B_{2,j}}\d A_{12,j}\int_{B_{2,j}\cap B_{3,j}}\d A_{23,j}^{\prime}\right)\prod_{j=1}^{n}e^{-\im\left[\bfS(A_{12,j},A_{23,j},A_{31,j})-\bfS(A_{12,j},A_{23,j}',A_{31,j})\right]}\\
 & \quad=\left(\prod_{j=1}^{n}\int_{B_{1,j}\cap B_{2,j}}\d A_{12,j}\right)\left(\prod_{j=1}^{n}\int_{B_{2,j}\cap B_{3,j}}\d A_{23,j}^{\prime}\right)\\
 & \qquad\times\exp\left[-\im\sum_{j=1}^{n}\left(\bfS(A_{12,j},A_{23,j},A_{31,j})-\bfS(A_{12,j},A_{23,j}',A_{31,j})\right)\right]\\
 & \quad=\int_{B_{1}\cap B_{2}}\d A_{12}\int_{B_{2}\cap B_{3}}\d A_{23}^{\prime}\ \exp\left[-\im\bfS(X_{1},X_{2},X_{3},X_{2}')\right]
\end{align*}
\end{proof}

\section{Free scalar field in curved spacetime}

\label{sec:Free-scalar-field}

\subsection{General scalar fields}

In this subsection, the spacetime $M$ is assumed to be an arbitrary
smooth manifold. Consider a general (possibly non-free) scalar field
$\phi$ on $M$. In this article we do not give a precise definition
of a general scalar field on $M$; instead it is roughly considered
as follows here. Since we want to consider the spectral projection
of the field operator $\phi(f)$ ($f\in C_{0}^{\infty}(M,\R)$, the
space of compactly supported smooth real functions), we assume that
$\phi(f)$ is represented as a self-adjoint operator on a Hilbert
space $\cH$, rather than as an element of an abstract {*}-algebra.
Furthermore we assume that $C_{0}^{\infty}(M)\ni f\mapsto\phi(f)$
is linear.

Let $\cD=\cD_{M}:=C_{0}^{\infty}(M)$, and $\cD'=\cD_{M}'$ denote
dual space of $\cD$; Roughly, $\cD'$ can be seen the space of Schwartz
distributions on $M$. Note that the general and precise definitions
of the distributions on manifolds can only be found rather scattered
throughout the literature, e.g., \cite[Chapter XVII]{Die1972}, \cite[Sec.6.3]{Hoer1990},
\cite[Sec.1.1]{BGP2007}; A unified presentation of the distributions
on manifolds is found in \cite[Chapter 3]{GKOS2001}. However, for
the time being, we do not need a precise definition of that notion
in this article; In fact, we may view $\cD'$ even as the \emph{algebraic}
dual of $\cD$ here. We shall not refer to the topologies on $\cD$
and $\cD'$ hereafter, and also to the Borel subsets of them; We will
consider only the Borel subsets of finite-dimensional linear (or affine)
spaces.

Let $f_{1},...,f_{n}\in\cD$ ($n\in\N$). For a Borel set $S\in\Borel(\R^{n})$,
define $D_{f_{1},...,f_{n}}(S)\subset\cD'$ by
\[
\SYM{D_{f_{1},...,f_{n}}(S)}{Dff}:=\left\{ F\in\cD'|(F(f_{1}),...,F(f_{n}))\in S\right\} .
\]
If the operators $\phi(f_{1}),...,\phi(f_{n})$ are pairwise commutative,
so that the simultaneous spectral resolution of them exists, the set
$D_{f_{1},...,f_{n}}(S)$ corresponds to the yes-no type measurement
that the value of the $\R^{n}$-valued observable $(\phi(f_{1}),...,\phi(f_{n}))$
is in $S$; Equivalently, $D_{f_{1},...,f_{n}}(S)$ can be seen as
a $\{0,1\}$-valued quantum observable. In this case, $D_{f_{1},...,f_{n}}(S)$
is called a \termi{$\phi$-observable subset} of $\cD'$. %
Let $\SYM{\cO_{\phi}}{Ophi}$ be the set of $\phi$-observable subsets
of $\cD'$. %
{} For $O\in\cO_{\phi}$, let $\SYM{E(O)}{E(O)}$ denote the corresponding
projection.

Let $f_{1},...,f_{n},g_{1},...,g_{m}\in\cD\setminus\{0\}$, and assume
that $\{f_{1},...,f_{n}\}$ is linearly independent, and $\phi(f_{1}),...,\phi(f_{n})$
are pairwise commutative; and similarly that $\{g_{1},...,g_{m}\}$
is linearly independent, and $\phi(g_{1}),...,\phi(g_{m})$ are pairwise
commutative. Let $S$ and $S'$ be bounded Borel sets of $\R^{n}$
and $\R^{m}$, respectively. Then we find that
\[
D_{f_{1},...,f_{n}}(S)=D_{g_{1},...,g_{m}}(S')\Then\Span\{f_{1},...,f_{n}\}=\Span\{g_{1},...,g_{m}\}\ (\text{and }n=m).
\]
Thus the \termi{domain} (or \termi{support}) of the $\phi$-observable
set $D_{f_{1},...,f_{n}}(S)\in\cO_{\phi}$ can be well-defined by
\[
\sfD(D_{f_{1},...,f_{n}}(S)):=\supp(f_{1})\cup\cdots\cup\supp(f_{n})=\bigcup_{f\in\Span\{f_{1},...,f_{n}\}}\supp(f).
\]
Then $O=D_{f_{1},...,f_{n}}(S)\in\cO_{\phi}$ can be seen as a $\{0,1\}$-valued
\emph{local} observable on the bounded spacetime domain $\sfD(O)$.
For general $\phi$-observable set $O\in\cO_{\phi}$, the suitable
definition of the domain $\sfD(O)$ of $O$ is less trivial; it will
be given as follows.%
{} For any subset $\cX$ of $\cD'$, define the subspaces $\Inv(\cX)$
of $\cD'$ and $\ann(\cX)$ of $\cD$ by
\[
\SYM{\Inv(\cX)}{Inv()}:=\{F\in\cD'|\forall r\in\R,\cX+rF=\cX\},
\]
\[
\SYM{\ann(\cX)}{ann()}:=\{f\in\cD|\forall F\in\cX,F(f)=0\}.
\]
Then we define the domain $\sfD(O)$ of $O\in\cO_{\phi}$ by
\[
\SYM{\sfD(O)}{D(O)}:=\bigcup\{\supp(f)|f\in\ann(\Inv(O))\}.
\]

\begin{lem}
\label{lem:D(D(S))<supp}If $f_{1},...,f_{n}\in\cD$ are linearly
independent, and $S\in\Borel(\R^{n})$, then
\[
\ann(\Inv(D_{f_{1},...,f_{n}}(S)))\subseteq\Span\{f_{1},...,f_{n}\},
\]
and hence
\[
\sfD(D_{f_{1},...,f_{n}}(S))\subseteq\supp(f_{1})\cup\cdots\cup\supp(f_{n}).
\]
The equalities hold if $S$ is bounded.
\end{lem}

\begin{proof}
Let $F(\vec{f}):=\left(F(f_{1}),...,F(f_{n})\right)$ and $G(\vec{f}):=\left(G(f_{1}),...,G(f_{n})\right)$.
Then we have
\begin{align*}
\Inv(D_{f_{1},...,f_{n}}(S)) & =\left\{ G\in\cD'|\forall r\in\R,\forall F\in\cD',\ F(\vec{f})\in S\iff F(\vec{f})+rG(\vec{f})\in S\right\} \\
 & =\left\{ G\in\cD'|\forall F\in\cD',\ F(\vec{f})\in S\then\forall r\in\R,F(\vec{f})+rG(\vec{f})\in S\right\} \\
 & =\left\{ G\in\cD'|S+\R G(\vec{f})=S\right\} .
\end{align*}
Therefore
\begin{align*}
\ann\left(\Inv(D_{f_{1},...,f_{n}}(S))\right) & =\left\{ h\in\cD|\forall G\in\cD',\ S+\R G(\vec{f})=S\then G(h)=0\right\} \\
 & \subseteq\left\{ h\in\cD|\forall G\in\cD',\ G(\vec{f})=0\then G(h)=0\right\} \\
 & =\Span\{f_{1},...,f_{n}\}.
\end{align*}
If $S$ is bounded, clearly the last ``$\subseteq$'' becomes ``$=$''.
\end{proof}

\subsection{Prior conditional probability for KG fields}

Recall the definitions in Section \ref{sec:Definitions-of-free}.
Let $\phi(f)$ ($f\in\cD=C_{0}^{\infty}(M)$) be a KG field on a globally
hyperbolic spacetime $M$, but we assume that each field operator
$\phi(f)$ is a selfadjoint operator on a Hilbert space $\cH$, rather
than an element of an abstract {*}-algebra. Let $\cN:=\Ker(\GreenOp)\subset\cD$. 

Our rough idea is that some part of the empirical laws of a KG field
can be expressed as the probabilistic laws on the prior conditional
probabilities 
\begin{equation}
\Prob_{\bOne}(O_{k+1},...,O_{n}|O_{1},...,O_{k}),\qquad O_{1},...,O_{n}\in\cO_{\phi}.\label{eq:P1(OO|OO)}
\end{equation}
For $f\in\cD$, let $\SYM{E_{f}(S)}{Ef()}$ ($S\in\Borel(\R)$) be
the spectral projection of $\phi(f)$. The above prior conditional
probability concerns only a finite number of field operators $\phi(f_{j})$
($j=1,...,N$). Hence consider the von Neumann algebra $\cA$ generated
by
\[
\bProj:=\left\{ E_{f_{j}}(S)|j=1,...,N,\ S\in\Borel(\R)\right\} .
\]
Except the special cases (e.g., where $f_{j}\in\cN$ for all $j$),
$\cA$ is a factor of type ${\rm I}_{\infty}$, and hence (\ref{eq:P1(OO|OO)})
will be defined for some suitable cases. The calculation of (\ref{eq:P1(OO|OO)})
will be reduced to the finite-dimensional cases in Section \ref{sec:Finite-dimensional-case}.
Of course, this idea will work well only on free fields, and possibly
we will encounter far more complicated situations on interacting fields.

Let $V$ be a finite-dimensional subspace of $\cD$ such that $(V,\GreenOp|_{V})$
is a symplectic vector space. In particular, $\dim V$ is even and
$V\cap\cN=\{0\}$. Note that $V$ is canonically identified with $(V+\cN)/\cN$
by $V\ni f\mapsto f+\cN$. For $f\in V$ and $f_{0}\in\cN$, let $[f+f_{0}]=[f+f_{0}]_{V}:=f$.
That is, the map $V+\cN\to V$, $h\mapsto[h]$, is the projection
from $V+\cN$ onto $V$.

For $f_{1},...,f_{n}\in V$ and $S\in\Borel(\R^{n})$, define $C_{f_{1},...,f_{n}}(S)=C_{f_{1},...,f_{n}}^{V}(S)\subset V$
by
\[
C_{f_{1},...,f_{n}}^{V}(S):=\{g\in V:\,(\GreenOp(g,f_{1}),...,\GreenOp(g,f_{n}))\in S\},
\]
similarly to Sec.~\ref{sec:Finite-dimensional-case}. Alternatively,
$C_{f_{1},...,f_{n}}^{V}(S)$ can be defined as a subset of $V^{*}$,
the dual space of $V$, canonically identified with $V$ w.r.t.~the
symplectic form $\GreenOp|_{V}$:
\[
C_{f_{1},...,f_{n}}^{V}(S):=\{F\in V^{*}:\,\left(F(f_{1}),...,F(f_{n})\right)\in S\}.
\]

For any subsets $\cX\subset\cD$ and $\cY\subset\cD'$, let
\[
\Ann(\cX):=\{F\in\cD':\forall f\in\cX,\,F(f)=0\},\qquad\cY|_{V}:=\{F|_{V}:F\in\cY\}.
\]

\begin{lem}
\label{lem:DcapAnn(N)}For $f_{1},...,f_{n}\in V+\cN$ ($n\in\N$)
and $S\in\Borel(\R^{n})$, we have
\[
\left(D_{f_{1},...,f_{n}}(S)\cap\Ann(\cN)\right)|_{V}=C_{[f_{1}],...,[f_{n}]}^{V}(S).
\]
\end{lem}

\begin{proof}
We have
\begin{align*}
 & \left(D_{f_{1},...,f_{n}}(S)\cap\Ann(\cN)\right)|_{V}\\
 & =\{F|_{V}:\ F\in\cD',\ (F(f_{1}),...,F(f_{n}))\in S,\ \forall f_{0}\in\cN,F(f_{0})=0\}\\
 & =\{F|_{V}:\ F\in\cD',\ (F([f_{1}]),...,F([f_{n}]))\in S,\ \forall f_{0}\in\cN,F(f_{0})=0\}\\
 & =\{F\in V^{*}:\ (F([f_{1}]),...,F([f_{n}]))\in S\}\\
 & =C_{[f_{1}],...,[f_{n}]}^{V}(S).
\end{align*}
\end{proof}
Define $\cO_{\phi,V}\subset\cO_{\phi}$ by
\[
\SYM{\cO_{\phi,V}}{OphiV}:=\left\{ D_{f_{1},...,f_{n}}(S)|S\in\Borel(\R^{n}),\ n\in\N,\ f_{1},...,f_{n}\in V+\cN,\ \GreenOp(f_{i},f_{j})=0,\ i,j=1,...,n\right\} .
\]
Recall that $\GreenOp(f,g)=0$ iff $\phi(f)$ and $\phi(g)$ are commutative.
Hence for any $O=D_{f_{1},...,f_{n}}(S)\in\cO_{\phi,V}$, the simultaneous
spectral resolution of the operators $\{\phi(f)|f\in\ann(\Inv(O))\}$
exists, since $\ann(\Inv(O))\subseteq\Span\{f_{1},...,f_{n}\}$ by
Lemma \ref{lem:D(D(S))<supp}. It follows that we can define the spectral
projection $\SYM{E(O)}{E(O)}$ corresponding $O\in\cO_{\phi,V}$.

An element of $\cO_{\phi,V}$ is called a \termi{$(\phi,V)$-observable subset}
of $\cD'$. By Lemma \ref{lem:DcapAnn(N)}, we have a canonical map
\[
\cO_{\phi,V}\to\cO(V),\qquad O\mapsto\SYM{O^{V}}{OV}:=\left(O\cap\Ann(\cN)\right)|_{V},
\]
where $\cO(V)$ is the set of observable subsets of $V$ defined in
Sec.~\ref{sec:Finite-dimensional-case}.

Recall the definition (\ref{eq:P1-finite-dim}) of $\Prob_{\bOne}$
in the finite-dimensional cases. For $1<k<n$ and $O_{1},...,O_{n}\in\cO_{\phi,V}$,
define the prior conditional probability $\Prob_{\bOne}^{V}$ w.r.t.~$O_{1},...,O_{n}$
by
\begin{align*}
\Prob_{\bOne}^{V}(O_{k+1},...,O_{n}|O_{1},...,O_{k}) & :=\Prob_{\bOne}(O_{k+1}^{V},...,O_{n}^{V}|O_{1}^{V},...,O_{k}^{V}),
\end{align*}
if the r.h.s.~is well-defined. 

This is also understood as follows. Let $\SYM{\cA_{V}}{AV}$ be the
von Neumann algebra generated by the projections $\{E(O)|O\in\cO_{\phi,V}\}$.
Then we see that $\cA_{V}$ is a factor of type ${\rm I}_{\infty}$,
that is, $\cA_{V}$ is isomorphic to $B(\cK)$, where $\cK$ is an
infinite-dimensional separable Hilbert space, and hence the usual
trace $\Tr$ on $B(\cK)$ is transferred to a trace $\Tr_{V}$ on
$\cA_{V}$. Thus we have
\[
\Prob_{\bOne}^{V}(O_{k+1},...,O_{n}|O_{1},...,O_{k})=\frac{\Tr_{V}(AB)^{*}AB}{\Tr_{V}A^{*}A},\qquad A:=E(O_{1})\cdots E(O_{k}),\ B:=E(O_{k+1})\cdots E(O_{n}).
\]

\subsubsection{The simplest (nontrivial) empirical law}

Let $f_{1},f_{2},f_{3}\in\cD$ satisfy $\GreenOp(f_{i},f_{j})\neq0$
for $i\neq j$. Let $V:=\Span\{f_{1},f_{2}\}$ and assume $f_{3}\in V+\cN$.
Let $S_{i}\in\Borel(\R)$ ($i=1,2,3$) be bounded and non-null, and
define the observable subsets $O_{i}\subset\cD'$ by
\[
O_{i}:=D_{f_{i}}(S_{i})=\{F\in\cD'|F(f_{i})\in S_{i}\},\qquad i=1,2,3.
\]
By Lemma \ref{lem:DcapAnn(N)} we have
\[
O_{i}^{V}=C_{[f_{i}]}^{V}(S_{i})=\{v\in V|\GreenOp(v,[f_{i}])\in S_{i}\}.
\]
Note that $[f_{i}]=f_{i}$ for $i=1,2$, and
\begin{equation}
[f_{3}]=\frac{\GreenOp(f_{3},f_{2})f_{1}+\GreenOp(f_{1},f_{3})f_{2}}{\GreenOp(f_{1},f_{2})}.\label{eq:=00005Bf3=00005D=00003D}
\end{equation}
In fact, let $f_{3}=a_{1}f_{1}+a_{2}f_{2}+f_{0},$ $a_{1},a_{2}\in\R$,
$f_{0}\in\cN$, then we have $\GreenOp(f_{1},f_{3})=a_{2}\GreenOp(f_{1},f_{2})$
and $\GreenOp(f_{2},f_{3})=-a_{1}\GreenOp(f_{1},f_{2})$, and hence
(\ref{eq:=00005Bf3=00005D=00003D}) follows. %

We will consider the explicit formula to obtain the simplest prior
conditional probability
\begin{align*}
\Prob_{\bOne}(f_{3},S_{3}|f_{1},S_{1};f_{2},S_{2}): & =\Prob_{\bOne}^{V}(O_{3}|O_{1},O_{2})=\Prob_{\bOne}(O_{3}^{V}|O_{1}^{V},O_{2}^{V})\\
 & =\frac{\Tr_{V}E_{f_{1}}(S_{1})E_{f_{2}}(S_{2})E_{f_{3}}(S_{3})E_{f_{2}}(S_{2})}{\Tr_{V}E_{f_{1}}(S_{1})E_{f_{2}}(S_{2})}.
\end{align*}

\begin{lem}
\label{lem:TrVEf1Ef2}Let $\Leb^{1}$ be the standard Lebesgue measure
on $\R^{1}$. Then
\begin{equation}
\Tr_{V}E_{f_{1}}(S_{1})E_{f_{2}}(S_{2})=\frac{\Leb^{1}(S_{1})\Leb^{1}(S_{2})}{2\pi|\GreenOp(f_{1},f_{2})|}.\label{eq:TrEf1Ef2}
\end{equation}
\end{lem}

\begin{proof}
Let $\Leb_{V}$ be the Lebesgue measure on $V$ normalized by the
volume form $(2\pi)^{-1}\GreenOp(\cdot,\cdot)|_{V\times V}$. By Proposition
\ref{prop:finite-dim-TrB1B2}, we have
\[
\Tr_{V}E_{f_{1}}(S_{1})E_{f_{2}}(S_{2})=\Leb_{V}(O_{1}^{V}\cap O_{2}^{V}).
\]
We see
\begin{align*}
O_{1}^{V}\cap O_{2}^{V} & =C_{f_{1},f_{2}}^{V}(S_{1}\times S_{2})=\{v\in V|\GreenOp(v,f_{i})\in S_{i},i=1,2\}\\
 & =\{a_{1}f_{1}+a_{2}f_{2}\in V|a_{2}\in-\GreenOp(f_{1},f_{2})^{-1}S_{1},\,a_{1}\in\GreenOp(f_{1},f_{2})^{-1}S_{2}\}.
\end{align*}
Note that $(2\pi)^{-1}|\GreenOp(f_{1},f_{2})|$ is the area of the
parallelogram formed by the four points $0,f_{1},f_{2},f_{1}+f_{2}\in V$,
measured by $\Leb_{V}$. It follows that if $\Leb^{1}(S_{1})=\Leb^{1}(S_{2})=1$,
then $\Leb_{V}(O_{1}^{V}\cap O_{2}^{V})=|\GreenOp(f_{1},f_{2})|^{-2}(2\pi)^{-1}|\GreenOp(f_{1},f_{2})|=(2\pi)^{-1}|\GreenOp(f_{1},f_{2})|^{-1}$.
Hence generally we have (\ref{eq:TrEf1Ef2}).
\end{proof}
Let $A\in V=\Span\{f_{1},f_{2}\}$. Then we see
\begin{equation}
A=\GreenOp(f_{1},f_{2})^{-1}\left[-\GreenOp(f_{2},A)f_{1}+\GreenOp(f_{1},A)f_{2}\right].\label{eq:A=00003DEff}
\end{equation}

\begin{lem}
\label{lem:S(A1A2A3)}Let $A^{(i)}\in V=\Span\{f_{1},f_{2}\},\ i=1,2,3$.
Then the signed area $\bfS(A^{(1)},A^{(2)},A^{(3)})$ of the triangle
$A^{(1)}A^{(2)}A^{(3)}$ measured by the bilinear form $\GreenOp$
is given by
\begin{equation}
\bfS(A^{(1)},A^{(2)},A^{(3)})=\frac{1}{2}\GreenOp(f_{1},f_{2})^{-1}\sum_{i=1}^{3}\left[\GreenOp(f_{1},A^{(i)})\GreenOp(f_{2},A^{(i+1)})-\GreenOp(f_{2},A^{(i)})\GreenOp(f_{1},A^{(i+1)})\right],\quad A^{(4)}:=A^{(1)}.\label{eq:S(A1A2A3)}
\end{equation}
\end{lem}

\begin{proof}
Eq.~(\ref{eq:S(A1A2A3)}) follows from the following equations:
\[
\bfS(A^{(1)},A^{(2)},A^{(3)})=\frac{1}{2}\left[\GreenOp(A^{(1)},A^{(2)})+\GreenOp(A^{(2)},A^{(3)})+\GreenOp(A^{(3)},A^{(1)})\right],
\]
\begin{equation}
\GreenOp(A^{(i)},A^{(j)})=\GreenOp(f_{1},f_{2})^{-1}\left[\GreenOp(f_{1},A^{(i)})\GreenOp(f_{2},A^{(j)})-\GreenOp(f_{2},A^{(i)})\GreenOp(f_{1},A^{(j)})\right],\label{eq:E(AiAj)}
\end{equation}
where (\ref{eq:E(AiAj)}) follows from (\ref{eq:A=00003DEff}).%
\end{proof}

Let $r_{1},r_{2},r_{3},r_{2}'\in\R$. Let $A_{ij}\equiv A_{ji}\in V=\Span\{f_{1},f_{2}\}$
($i,j=1,2,3,i\neq j$) be such that $\GreenOp(f_{i},A_{ij})=r_{i},\ \GreenOp(f_{j},A_{ij})=r_{j}$.
Explicitly we have
\[
A_{12}=\GreenOp(f_{1},f_{2})^{-1}\left(-r_{2}f_{1}+r_{1}f_{2}\right),
\]
from (\ref{eq:A=00003DEff}).

\begin{lem}
\label{lem:E(f1,A23)}We have
\[
\GreenOp(f_{1},A_{23})=e_{32}^{-1}\left(e_{12}r_{3}-e_{13}r_{2}\right),\quad\GreenOp\left(f_{2},A_{13}\right)=e_{31}^{-1}\left(e_{21}r_{3}-e_{23}r_{1}\right),\qquad e_{ij}:=\GreenOp(f_{i},f_{j}).
\]
\end{lem}

\begin{proof}
Directly follows from $\GreenOp([f_{3}],A_{13})=r_{3}$, $\GreenOp([f_{3}],A_{23})=r_{3}$
and (\ref{eq:=00005Bf3=00005D=00003D}).
\end{proof}

\begin{lem}
\label{lem:S(A12,A23,A31)}We have
\[
\bfS(A_{12},A_{23},A_{31})=\frac{\left(e_{23}r_{1}+e_{31}r_{2}+e_{12}r_{3}\right)^{2}}{2e_{12}e_{23}e_{31}},\qquad e_{ij}:=\GreenOp(f_{i},f_{j}).
\]
\end{lem}

\begin{proof}
By Lemma \ref{lem:S(A1A2A3)}, \ref{lem:E(f1,A23)}, and straightforward
calculations. %
\end{proof}

Let $A_{23}'\in V$ satisfy $\GreenOp(f_{2},A_{23}')=r_{2}',\ \GreenOp([f_{3}],A_{23})=r_{3}$.
Let %
$\vec{r}:=(r_{1},r_{2},r_{3},r_{2}')$, and
\[
\bfS(\vec{r}):=\bfS(A_{12},A_{23},A_{32})-\bfS(A_{12},A_{23}',A_{32}).
\]
By Lemma \ref{lem:S(A12,A23,A31)}, explicitly we can write
\begin{equation}
\bfS(\vec{r})=\frac{1}{2e_{12}e_{23}e_{31}}\left[\left(e_{23}r_{1}+e_{31}r_{2}+e_{12}r_{3}\right)^{2}-\left(e_{23}r_{1}+e_{31}r_{2}'+e_{12}r_{3}\right)^{2}\right].\label{eq:S(r)=00003D}
\end{equation}
For fixed $f_{1},f_{2},f_{3}$, the value of $\bfS(\vec{r})$ is determined
by $A_{12}$ and $A_{23}'$, since 
\begin{equation}
r_{1}=\GreenOp(f_{1},A_{12}),\ r_{2}=\GreenOp(f_{2},A_{12}),\ r_{3}=\GreenOp(f_{3},A_{23}'),\ r_{2}'=\GreenOp(f_{2},A_{23}').\label{eq:E(f1,A12)=00003Dr1}
\end{equation}
Hence the following expression makes sense:
\[
\int_{\cX}\d A_{12}\int_{\cY}\d A_{23}'\ e^{-\im\bfS(\vec{r})},\qquad.
\]
where $\cX,\cY\subset V$, and $\d A$ is the Lebesgue measure on
$V$ normalized by the volume form $(2\pi)^{-1}\GreenOp(\cdot,\cdot)|_{V\times V}$.
Let
\[
O_{i}:=D_{f_{i}}(S_{i})=\{F\in\cD'|F(f_{i})\in S_{i}\},\qquad i=1,2,3.
\]
By Lemma \ref{lem:DcapAnn(N)}, we see
\[
O_{i}^{V}=C_{[f_{i}]}^{V}(S_{i})=\{v\in V|\GreenOp(v,[f_{i}])\in S_{i}\}.
\]
Note that $[f_{i}]=f_{i}$ for $i=1,2$. Then by Proposition \ref{prop:Q(B1B2B3)-explicit}
we have
\[
\Tr_{V}E_{f_{1}}(S_{1})E_{f_{2}}(S_{2})E_{f_{3}}(S_{3})E_{f_{2}}(S_{2})=\int_{O_{1}^{V}\cap O_{2}^{V}}\d A_{12}\int_{O_{2}^{V}\cap O_{3}^{V}}\d A_{23}'\ e^{-\im\bfS(\vec{r})}.
\]
Thus by Lemma \ref{lem:TrVEf1Ef2},%
{} we obtain the following explicit formula, which is considered as
one of the simplest empirical laws of the KG field.
\begin{thm}
\label{thm:KG3}Let $f_{1},f_{2},f_{3}\in\cD$ satisfy $\GreenOp(f_{i},f_{j})\neq0$
for $i\neq j$. Let $V:=\Span\{f_{1},f_{2}\}$ and assume $f_{3}\in V+\cN$.
Let $S_{i}\in\Borel(\R)$ ($i=1,2,3$) be bounded and non-null. Then
we have
\[
\Prob_{\bOne}(f_{3},S_{3}|f_{1},S_{1};f_{2},S_{2})=\frac{2\pi|\GreenOp(f_{1},f_{2})|}{\Leb^{1}(S_{1})\Leb^{1}(S_{2})}\int_{O_{1}^{V}\cap O_{2}^{V}}\d A_{12}\int_{O_{2}^{V}\cap O_{3}^{V}}\d A_{23}'\ e^{-\im\bfS(\vec{r})},
\]
where $\bfS(\vec{r})$ is given by (\ref{eq:S(r)=00003D}) and (\ref{eq:E(f1,A12)=00003Dr1}).
Of course, we can replace the above $e^{-\im\bfS(\vec{r})}$ with
$e^{\im\bfS(\vec{r})}$ or $\cos\bfS(\vec{r})$.
\end{thm}

\providecommand{\noopsort}[1]{}\providecommand{\singleletter}[1]{#1}%

\end{document}